\documentclass[onecolumn,12pt]{IEEEtran}
\usepackage{graphicx}  % Written by David Carlisle and Sebastian Rahtz

\usepackage{subfigure} % Written by Steven Douglas Cochran
\usepackage{stfloats}  % Written by Sigitas Tolusis
\usepackage{amsmath,amsfonts,amssymb,amscd}  % From the American Mathematical Society
\usepackage[ruled,vlined]{algorithm2e} % For writing the algorithm

\usepackage{graphicx}
\usepackage{color}

\newtheorem{remark}{\bfseries Remark}
\newtheorem{theorem}{\bfseries Theorem}
\newtheorem{lemma}{\bfseries Lemma}

\newcommand{\qed}{\hfill \ensuremath{\Box}}

\newcommand{\Ex}{\mathop{\bf E\/}}

\usepackage{pdfsync}
\usepackage{epstopdf}
\usepackage{setspace}

\title{\LARGE{Distributed Algorithms for Consensus and Coordination in the Presence of Packet-Dropping Communication Links\\
\Large{Part I: Statistical Moments Analysis Approach }}}
 \author{\authorblockN{Alejandro~D.~Dom\'{i}nguez-Garc\'{i}a,~\IEEEmembership{Member,~IEEE}}\\
\authorblockN{Christoforos N. Hadjicostis,~\IEEEmembership{Senior Member,~IEEE}}\\
\authorblockN{Nitin H. Vaidya,~\IEEEmembership{Fellow,~IEEE}} \\
\authorblockN{~} \\
\authorblockN{September 28, 2011}
\thanks{University of Illinois at Urbana-Champaign.  Coordinated Sciences Laboratory technical report UILU-ENG-11-2207 (CRHC-11-05)}
\thanks{A. D. Dom\'{i}nguez-Garc\'{i}a and N. H. Vaidya are with the Department of Electrical and Computer Engineering at the University of Illinois at Urbana-Champaign, Urbana, IL 61801, USA. E-mail:  \{aledan, nhv\}@ILLINOIS.EDU.}
\thanks{C. N. Hadjicostis is with the Department of Electrical and Computer Engineering at the University of Cyprus, Nicosia, Cyprus, and also with the Department of Electrical and Computer Engineering at the University of Illinois at Urbana-Champaign, Urbana, IL 61801, USA. E-mail:  chadjic@UCY.AC.CY.}
\thanks{The work of A. D. Dom\'{i}nguez-Garc\'{i}a was supported in part by  NSF under Career Award ECCS-CAR-0954420. The work of C. N. Hadjicostis was supported in part by the  
European Commission (EC) Seventh Framework Programme (FP7/2007-2013) 
under grant agreements INFSO-ICT-223844 and PIRG02-GA-2007-224877. The work of N. H. Vaidya was supported in part by Army Research Office grant W-911-NF-0710287 and NSF award 1059540. Any opinions, findings, and conclusions or recommendations expressed here are those of the authors and do not necessarily reflect the views of the funding agencies, the U.S. government, or the European Commission.}

}

\begin{document}

\markboth{Coordinated Sciences Laboratory technical report UILU-ENG-11-2207 (CRHC-11-05)}{}

\maketitle

\begin{abstract}
This two-part paper discusses robustification methodologies for linear-iterative distributed algorithms for consensus and coordination problems in multicomponent systems, in which unreliable communication links may drop packets. We consider a setup where  communication links between  components can be asymmetric (i.e., component $j$ might be able to send information to component $i$, but not necessarily vice-versa), so that the information exchange between components in the system is in general described by a directed graph that is assumed to be strongly connected. In the absence of communication link failures, each component $i$ maintains two auxiliary variables and updates each of their values to be a linear combination of their corresponding previous values and the corresponding previous values of neighboring components (i.e., components that send information to node $i$). By appropriately initializing these two (decoupled) iterations, the system components can asymptotically calculate variables of interest in a distributed fashion; in particular, the average of the initial conditions can be calculated as a function that involves the ratio of these two auxiliary variables. The focus of this paper  to robustify this double-iteration algorithm against communication link failures. We achieve this by modifying the double-iteration algorithm (by introducing some additional auxiliary variables) and prove that the modified double-iteration converges almost surely to average consensus. In the first part of the paper, we study the first and second moments of the two iterations, and use them to establish convergence, and illustrate the performance of the algorithm with several numerical examples. In the second part, in order to establish the convergence of the algorithm, we use coefficients of ergodicity commonly used in analyzing inhomogeneous Markov chains.

\end{abstract}
\newpage

\onehalfspace

\section{Introduction}
The design of protocols and algorithms for distributed computation and control/decision tasks has attracted significant attention by the computer science, communication, and control communities (e.g., \cite{Ly:96, OlFa:07, Koetter03, Rabbat04, Giridhar05, Hromkovic05, Ts:84} and references therein). For example, given i) a collection of robots moving in the plane, ii) a collection of sensors in a sensor network, or iii) a collection of distributed energy resources in an electrical grid, the components may be interested in, respectively, i) agreeing on a common direction to follow (this common direction could be provided by a leader robot), ii) measurement averaging (with each sensor providing a local measurement of a global quantity), or iii) collectively providing a predetermined total amount of active power subject to the constraints of each distributed resource. In the control literature, the first and second problems are respectively known as consensus and  average consensus (see, e.g., \cite{OlFa:07}), whereas the third problem can be considered as a distributed resource coordination problem \cite{DoHa:11c,DoHa:11b}. 

In this two-part paper, we consider multicomponent systems in which each component  can exchange information with other components in its neighborhood in order to compute, in a distributed fashion, some quantity of interest.  In our setup, communication links between  components (nodes) can be asymmetric (i.e., component $j$ might be able to send information to component $i$, but not necessarily vice-versa), a situation that arises in a wireless setting if the transmission power available to different nodes are also different. In this setting, the information exchange between components in the system can be described by a directed graph which is assumed to be strongly connected. Through an iterative process, nodes in the network are required to compute (using only information made available by their neighbors) the quantity of interest. In particular, we study linear-iterative algorithms in which each node $j$ maintains a value (or a set of values) that is updated to be a weighted linear combination of node $j$'s own previous value and the previous values of its neighboring nodes (i.e., nodes that  transmit information to node $j$). The main focus of the paper is to develop strategies to robustify the linear-iterative algorithms described above against communication links that may drop packets. 

In the context of consensus  and average-consensus problems,  an extensive literature in the control community  focuses on the linear-iterative algorithms described above (e.g., \cite{Ts:84, Ja:03, OlFa:07,SundaramHadjicostis08, XiaoBoyd04, Moreau05,OlsTsi:09}  and references therein).  These works have revealed that if the network topology satisfies certain conditions, the weights for the linear iteration can be chosen so that all   the nodes asymptotically converge to the same value (even if the network connections are time-varying). Additionally, if the interconnection topology is invariant and bidirectional (i.e., if node $j$ can send information to node $i$, then node $i$ can send information to node $j$), simple techniques can be used to choose the weights of the linear iteration  so as to ensure that, after running the linear iteration, the nodes will asymptotically reach consensus to the average of their initial values \cite{OlFa:07,SundaramHadjicostis08, XiaoBoyd04}.  Other works have looked at the consensus and average-consensus problems when the interconnection topology is described by a {\em directed} graph. In particular, the authors of \cite{SaMu:2003}  focus on  continuous-time linear iterations and state necessary and sufficient conditions for a network of integrators to asymptotically reach agreement to a common value (but not necessarily the average of their initial values). Similarly, the authors of \cite{XiaoBoyd04} consider discrete-time iterations, and provide necessary and sufficient conditions on the weights that allow the nodes to asymptotically reach consensus to the average of their initial values. Additionally, the work in \cite{DoHa:11d,DoHa:11a} discusses how average-consensus can be reached asymptotically with linear-iterative algorithms in which the nodes use fixed weights in their linear updates and also develops linear-iterative algorithms where the nodes adapt their weights in a distributed fashion so that asymptotically average-consensus is reached. In the context of resource coordination, there is some recent work \cite{DoHa:10, DoHa:11c,DoHa:11b} that also focuses on linear-iterative algorithms, similar to those used to address consensus and average-consensus problems. Other recent work has addressed the related problem of achieving consensus and average-consensus in a multicomponent system where some nodes can exhibit malicious behavior   \cite{SuHa:11,PaBi:11}. These works assume fault-free communication links, but are related to what we do in this paper in the sense that they can be used to handle unreliable nodes (as opposed to links).

In our development, we adopt a very general model for the communication modality between nodes, which allows asymmetric information structures, in the sense that if node $i$ can transmit information to another node $j$, it is not necessarily true that node $j$ can transmit information to node $i$.  We only require that each node, apart from seeing incoming transmissions sent to it by neighboring nodes, knows the number of nodes that it can transmit information to, which in graph-theoretic terms is referred to as the out-degree of that node.  In fact, in the proposed algorithm, each node will broadcast the same quantity to all receiving nodes, which simplifies the communication scheme between sending and receiving nodes (as it is not necessary for each sending node to separately communicate with each receiving node). 

When the communication network is perfectly reliable (no packet drops),  the collective dynamics of the linear iterations can be described by a discrete-time transition system with no inputs in which the transition matrix is column stochastic and primitive. Then, each node will run two identical copies of the linear iteration each of which, however is initialized differently depending on the problem to be solved.In this paper we mostly focus on the average consensus problem. Under proper initialization, it can be shown that each node will asymptotically calculate the desired value as a function of the outcomes of the two iterations. The details of these double-iteration approach are provided in \cite{DoHa:11d,DoHa:11a} for the average consensus case and in \cite{DoHa:11c,DoHa:11b} for the resource coordination problem. For the average-consensus problem, the double-iteration algorithm is a particular case of the algorithm in \cite{Benezit:10} (which is a generalization of the algorithm proposed in  \cite{KeDo:03}), where the matrices describing each linear iteration are allowed to vary as time evolves, whereas in our setup (for the ideal case when there are no communication link failures) the  transition matrix is fixed over time. 

The focus of this paper is to robustify the double-iteration algorithm (informally described above and formally described in Section~\ref{preliminaries}) so that it can tolerate failures in communication links and converge to the average value. Our communication link reliability model assumes that at each time step, a communication link is unavailable with some probability. In other words, a packet containing information from node $i$ to node $j$ is dropped with some probability. Next we informally describe our robustification approach. Consider two nodes $i$ and $j$, and assume that $j$ receives information from node $i$ but not necessarily vice-versa. Let us refer to $j$ as the receiving node (or receiver) and $i$ as the sending node (or sender). An important requirement is for the graph describing the communication network to be strongly connected, which implies every node must be able to act  both as a sender and as a receiver. Then, for each of the two iterations each node performs, node $i$ (the sender) will keep track of the following quantities of interest: i) its own internal state (as captured by the state variables maintained in the double iteration scheme of \cite{Benezit:10,DoHa:10}; ii) the total mass broadcasted so far (to be described in detail soon); and  iii) the total received mass from each node $l$ that sends information to node $i$. Similarly, for both iterations, each node $j$ (the receiver) updates the value 
of its internal state to be a  linear combination of its own previous internal state value (weighted by the inverse of the number of nodes that have $j$ as a neighbor) and the difference between the two most recently  received mass values from each of its neighbors (also weighted by the inverse of the number of nodes that have $j$ as a neighbor). At time instant $k$, the total broadcasted mass by node $j$ is the sum up to  (and including) time step $k$ of the sequence of values of node $j$'s internal  value, weighted by the inverse of the number of nodes that receive values  from node $j$). Additionally, node $j$ updates the value of the  received mass from each node $l$ that sends information to node $j$ as  follows: the received mass from node $l$ is the total broadcasted mass sent  by node $l$ up to time $k$ if the communication link from node $l$ to node $j$ is available at time step $k$; otherwise, the received mass remains the same as the most recently received mass from node $l$. An implicit assumption here is that messages broadcasted by node $l$ are tagged with the sender's identity so that the receiving node $j$ can determine where different packages have originated from.

Recent work that has addressed the consensus and average-consensus problems in the presence of unreliable communication links \cite{PaBa:07,FaZa:09,ChTr:10} has employed a  communication link availability model  similar to ours. The work in \cite{PaBa:07} assumes that the graph describing the communication network is undirected and that when a communication link fails it affects communication in both directions. Additionally, nodes have some mechanism to detect link unavailability and compensate for it by rescaling their other weights (so that the resulting transition matrix remains column stochastic). Following this strategy, the authors show asymptotic convergence to the average of initial conditions and also calculate the rate at which the variance of the total deviation from the average converges to zero. The work in \cite{FaZa:09} does not require   the graph describing the communication network to be undirected and proposes two compensation methods to account for communication link failures. In the first method, the so-called biased compensation method,  the receiving node   compensates for the unavailability of an incoming link by adding the weight associated to the unavailable link to its own weight (so that the resulting matrix remains row stochastic). In the second method, called the balanced compensation method, the receiving node  compensates for link unavailability by rescaling all the incoming link weights so that the resulting matrix remains row stochastic. The key in both methods is the fact that at each time step, the resulting weight matrix is row stochastic; the authors  show that the nodes  converge almost surely to the same value, but this value is not necessarily the average of the initial conditions. The work in  \cite{ChTr:10}, which does not require the  communication graph to be undirected, proposes a correction strategy that corrects the errors in the quantity (state)   iteratively calculated by each node, so that the nodes obtain the average of their  initial values. This correction strategy is based on each node maintaining some auxiliary variable that accounts for the amount by which node $i$ changes its state due to the updates from its neighbors, i.e., the nodes that can send information to node $i$. For their strategy to work and ensure that the nodes converge almost surely to the average consensus, the authors rely on the nodes sending acknowledgment messages and retransmitting information an appropriate number of times. 

In \cite{Benezit:10}, the authors  proposed a gossip-based algorithm for average-consensus ver a directed graph where the transition matrices describing the nodes' collective dynamics change at every iteration step (depending on which node awakes). This scheme requires the node that is awake to perform an internal state update and send  its internal state (weighted by the corresponding out-going link weight) to its neighbors. This approach results in generates a sequence of column stochastic matrices (not necessarily primitive) with the property that all the diagonal entries remain positive. The authors prove that by running two such iterations in parallel, one of them initialized with the values on which the average operation is to be performed and the other with the all-ones vector, each node  will asymptotically achieve average consensus by taking the ratio of the two values in maintains. A key premise in their proof is that column stochasticity of the transition matrix is maintained over time, which requires sending nodes to know the number of nodes that are listening. This suggests that i) either the communication links are perfectly reliable, or ii) there is some acknowledgment and retransmission mechanism that ensures messages are delivered to the listening nodes at every round of information exchange. In this paper, we remove such assumptions and robustify the double-iteration algorithm against unreliable communication links using  a pure broadcast-message model without any requirement for an acknowledgment/retransmission mechanism. Thus, despite the reliance  of our algorithm on the ratio of two linear iterations, it is different both in the communication model we assume---a broadcast model in our case---and also in the nature of the protocol itself---our focus is on ensuring convergence in the presence of communication link failures. 

An additional assumption made in  \cite{Benezit:10} is that the diagonal entries of the transition matrix (at every step) remain positive. In our model, we originally consider that nodes do not drop self-packets. However, to ease the analysis, we remove this assumption and consider the case where self-packet drops are also allowed at every time step, which i) allows us to handle  intermittent faults in the node processing device, and ii)  removes the assumption that all diagonal entries must be positive at every step. Finally, the analysis machinery in \cite{Benezit:10} is quite different from the one used in this paper. We employ moment analysis of the two iterations to establish that they are linearly related as the number of steps goes to infinity, while  \cite{Benezit:10} relies on establishing weak ergodicity of the product of the transition matrices as the number of steps goes to infinity. Finally, as it will be shown in the second part of this paper, our algorithm can be re-casted into a similar framework as the one in \cite{Benezit:10} by augmenting the dimension of the vector describing the collective dynamics to account for the packets that get dropped once there is a communication failure. Note, However, that the resulting matrices will be column stochastic but will not necessarily have  strictly positive entries on their diagonals. In the second part of the paper, we provide an analysis framework to establish the convergence of our algorithm and generalizes the ideas in \cite{Benezit:10} to the case when the matrices describing the system collective dynamics do not have strictly positive diagonals. In this regard, we will show that even in the case where self-packet drops are not allowed, the resulting transition matrices might still have zero diagonal entries.

The remainder of this paper is organized as follows. Section~\ref{preliminaries} provides background on graph theory, introduces the communication model, and briefly describes the non-robust version of the double-iteration algorithm we use in this work. Section~\ref{algorithm} describes the proposed strategy to robustify the double-iteration algorithm against communication link failures and illustrates the use/performance of the algorithms via several examples. The convergence analysis of the robustified double-iteration algorithm is provided in Section~\ref{convergence_analysis}. Concluding remarks are presented in Section~\ref{concluding_remarks}.

\section{Preliminaries} \label{preliminaries}
This section provides background of  graph-theoretic notions that are used to describe  the communication network and the distributed consensus/coordination setup,   introduces the basic communication link availability model, and  reviews a previously proposed  two-iteration algorithm that can be used to solve the class of problems addressed in this paper  when the communication network is perfectly reliable.

%formulation of the basic algorithm
%
%
% to solve the class problems addressed in this paper. that assumes a perfectly reliable communication network

\subsection{Network Communication Model}
Let discrete time instants be indexed $k=0,1,\dots$; then, the information exchange between nodes (components) at each time instant $k$ can be described by a directed graph $\mathcal{G}[k]=\{\mathcal{V}, \mathcal{E}[k]\}$, where $  \mathcal{V}=\{1,2,\dots,n\} $ is the vertex set (each vertex---or node---corresponds to a system component), and $  \mathcal{E}[k] \subseteq \mathcal{V} \times \mathcal{V} $ is the set of edges, where $(j,i) \in \mathcal{E}[k]$ if node $j$ can receive information from node $i$ at instant $k$. It is assumed that $\mathcal{E}[k] \subseteq \mathcal{E},~\forall k \geq 0$, where $\mathcal{E}$ is the set of edges that describe all possibly available communication links between nodes; furthermore, the graph $(\mathcal{V},\mathcal{E})$ is  assumed to be strongly connected. All nodes that can possibly transmit information to node $j$  are called its  in-neighbors, and are represented by the set $\mathcal{N}_j^{-}=\{i \in \mathcal{V}: (j,i) \in \mathcal{E} \}$. Note that there are  self-loops for all nodes in  $\mathcal{G}$ (i.e., $(j,j) \in \mathcal{E} $ for all $j \in \mathcal{V}$).   The number of neighbors of $j$ (including itself) is called the in-degree of  $j$ and is denoted by $\mathcal{D}_j^-=|\mathcal{N}_j^-|$. The  nodes that have $j$ as neighbor (including itself) are called its out-neighbors and  are denoted by $\mathcal{N}_j^+=\{l \in \mathcal{V}: (l,j) \in \mathcal{E} \}$; the out-degree of node $j$ is $\mathcal{D}_j^+=|\mathcal{N}_j^+|$.

The existence of a communication link  from node $i$ to node $j$ can be described in probabilistic terms as follows.
% At instant $k$, let $\ell_{ji}[k]$ be  an indicator variable that takes value 1 when there is a communication link from node $i$ to node $j$,  and zero otherwise:
%\begin{align}
%& \ell_{ji}[k]=\left\{\begin{array}{cc} 1, &   \text{if } (j,i) \in \mathcal{E}[k],\\
%0, &\text{if } (j,i) \notin \mathcal{E}[k].\end{array}\right. \label{link_ji}
%\end{align}
%Then, if the probability of link existence is $p$, the pmf of $\ell_{ji}[k]$ is given by
%\begin{align}
%& \Pr\{\ell_{ji}[k]=m\}=\left\{\begin{array}{cc} q, & \text{if } m=1,\\
%1-q, & \text{if } m=0.\end{array}\right. \label{Bernouilli_model}
%\end{align}
%
At instant $k$, let $x_{ji}[k],~\forall i,j \in \mathcal{V}$ be an indicator variable that takes value 1 with probability $q$ and takes value zero with probability $1-q$, i.e.,
\begin{align}
& \Pr\{x_{ji}[k]=m\}=\left\{\begin{array}{cc} q, & \text{if } m=1,\\
1-q, & \text{if } m=0.\end{array}\right. \label{Bernouilli_model}
\end{align}
Then, for all $k \geq 0$, the existence of a communication link between nodes $i$ and $j$  can be described be another indicator variable $\ell_{ji}[k]$ defined as
\begin{align}
& \ell_{ji}[k]=\left\{\begin{array}{cc} x_{ji}[k], &   \text{if } (j,i) \in \mathcal{E},\\
0, &\text{if } (j,i) \notin \mathcal{E}.\end{array}\right. \label{link_ji}
\end{align}
It follows that  $\mathcal{E}[k]$ contains the elements of $\mathcal{E}$ for which $\ell_{ji}[k]=x_{ji}[k]=1$.

\subsection{Double-Iteration Algorithm Formulation in Perfectly Reliable Communication Networks} \label{non_robustified_algorithm}
When the communication network  of a multi-component system is perfectly reliable, i.e.,  
$\Pr\{\ell_{ji}[k]=1\}=1, ~\forall  (j,i) \in \mathcal{E},~\forall k \geq 0$, it was shown in \cite{DoHa:11a,DoHa:11b} that the components of the multi-component system can asymptotically solve average consensus  and  resource coordination problems in a distributed fashion by running two separate appropriately initialized  linear iterations  of the form
%This algorithm  uses the same set of constant weights for iterations \eqref{general1} and \eqref{general2};  specifically, each node  $j$  sets its weights to be $p_{jj}=q_{jj}=1/(\mathcal{D}_j^+)$ and $p_{ij}=q_{ij}=1/(1+ \mathcal{D}_j^+)$ for all $i$ that have $j$ as a neighbor, i.e., all $i$ such that $j\in \mathcal{N}_i$. As a result, each node updates its value as
\begin{align}
& y_j[k+1]=\sum_{i \in \mathcal{N}_j^-} \frac{1}{\mathcal{D}_i^+}y_i[k], \label{splitting11} \\
&  z_j[k+1]=\sum_{i \in \mathcal{N}_j^-} \frac{1}{\mathcal{D}_i^+}z_i[k], \label{splitting21}
\end{align}
where $\mathcal{D}_j^+$ ($\mathcal{D}_i^+$) is the out-degree of node  $j$ ($i$). A requirement in all cases is that the underlying communication graph $(\mathcal{G},\mathcal{E})$ is strongly connected.  

\subsubsection{Average Consensus Problem} In this problem, the nodes aim to obtain the average of the values $v_j,~j=1,\dots,n$, they each posses. In  \cite{DoHa:11a}, it was shown that if the initial conditions in \eqref{splitting11} (referred to as iteration 1) are set to $y_j[0]=v_j$, and the initial conditions in  \eqref{splitting21} (referred to as iteration 2) are set to $z_j[0]=1$, then the nodes can asymptotically calculate $\overline{v}:=\sum_{j=1}^n v_j/n$ as
\begin{align}
\overline{v}=\lim_{k \rightarrow \infty}\frac{y_j[k]}{z_j[k]},
\end{align}
by running the two iterations  in \eqref{splitting11} and \eqref{splitting21}.

\subsubsection{Resource Coordination Problem} In this problem, each node $j$ can contribute
a certain amount $\pi_j \geq 0$ of a given resource, which is upper and lower bounded by known capacity limits $\pi_j^{max}$ and $\pi_j^{min}$ respectively. The challenge is to coordinate the components so that they collectively provide a pre-determined total amount $\rho_d= \sum_{j=1}^n \pi_j$ of the resource\footnote{In the development in \cite{DoHa:10,DoHa:11b}, it is assumed that $ \sum_{j=1}^n \pi_j^{min} \leq \rho_d  \leq \sum_{j=1}^n \pi_j^{max}$;  this is not a restrictive assumption because in the proposed algorithms, all nodes will be able to know if $\rho_d  < \sum_{j=1}^n \pi_j^{min}$ or if $\rho_d> \sum_{j=1}^n \pi_j^{max}$ (which means that no feasible solution exists).} as specified by an external ``leader." In \cite{DoHa:11b}, it was shown that i) if the initial conditions in \eqref{splitting11} are set to $y_j[0]=\rho_d/m-\pi_j^{min}$ if $j$ is an out-neighbor of the leader (where $m \geq 1$ is the number of nodes contacted initially by the external leader) and $y_j[0]=-\pi_j^{min}$ otherwise, and ii) if the initial conditions in  \eqref{splitting21} are set to $z_j[0]=\pi_j^{max}-\pi_j^{min}$, then the nodes can asymptotically calculate their own resource contribution $\pi_j$ as
\begin{align}
& \pi_j:=\lim_{k \rightarrow \infty} \big( \pi_j^{min}+ \frac{y_j[k]}{z_j[k]}(\pi_j^{max}-\pi_j^{min})\big)=\pi_j^{min}+  \frac{\rho_d-\sum_{l=1}^n \pi_l^{min}}{\sum_{l=1}^n\ell_l} (\pi_j^{max}-\pi_j^{min}), \label{splitting3}
\end{align}
which satisfies 
\begin{align}
&\pi_j^{min} \leq \pi_j \leq \pi_j^{max},~\forall j, \nonumber\\
&  \sum_{j=1}^n \pi_j=\rho_d. \label{constraints}
\end{align}

In this paper, we start with a double iteration of the form in \eqref{splitting11}--\eqref{splitting21} that is used for either average consensus or coordination, and develop systematic methodologies to handle packet drops in the communication links.

\section{Robustification of Double-Iteration Algorithm} \label{algorithm}
In this section, the algorithm described in Section~\ref{non_robustified_algorithm} is modified so as to make it robust against communication link failures. As in \eqref{splitting11}--\eqref{splitting21}, each node will run two iterations (which we  refer to as iterations 1 and 2) to calculate  quantities of interest and eventually  solve the average consensus or coordination problems.

Consider the setup described in the previous section: we are given a (possible directed) strongly connected graph $(\mathcal{G},\mathcal{E})$ representing a multicomponent system and its communication links between its components.  For the sake of generality, let us refer to $j$ as the receiving node (or receiver) and $i$ as the sending node (or sender). For each of the two iterations, node $i$ (the sender) will calculate several quantities of interest, which we refer to as: i) internal state; ii) total broadcasted mass; and  iii)  total received mass from each in-neighbor $l$ of node $i$, i.e., for each node $l \in \mathcal{N}_i^-$.  For both iterations 1 and 2, each node $j$ updates the value of its internal state to be a  linear combination of its own previous internal state value (weighted by the inverse of the number of nodes that have $j$ as a neighbor, i.e., $1/\mathcal{D}_j^+$) and the sum  (over all its in-neighbors) of the difference between the two most recently  received mass values. At instant time $k$,  the total broadcasted mass   is the sum up to (and including) step $k$ of the weighted value of node $j$'s internal state (used to update the  internal state of node $j$). Additionally, node $j$ (the receiver) updates the value of the received mass from node $l$ to be either the total broadcasted mass sent by node $i$ if the communication link from  $i$ to $j$ is available at instant $k$, or the most recently received mass value from node $i$, otherwise.  An implicit assumption here is that messages broadcasted  by node $i$ are tagged with the sender's identity so that the receiving node $j$ can determine where messages originated from.

%The idea here is that at every time step $k$, a node $i$ nodes calculates and broadcast their total mass broadcasted  up to and including time step $k$. When a communication link is present between node $i$ and $j$, receiving node $j$ will use the total mass value broadcasted  by node $i$ to update the value of some quantity of interest that $j$ needs to calculate. If a communication link is not present between $i$ and $j$, then the receiving node $j$ will use the most updated broadcasted mass sent by node $i$ in the calculations.

For iteration 1, let $y_j[k]$ denote node $j$'s internal state at time instant $k$, $\mu_{lj}[k]$ denote the mass broadcasted from node $j$ to each of its out-neighbors $l$  (this is a single value ad it is the quantity is the same for each out-neighbor $l$ of node $j$, i.e., for each $l \in \mathcal{N}_j^-$), and $\nu_{ji}[k]$ denote  the total mass received at node $j$ from node $i \in \mathcal{N}_j^-$. Similarly, let $z_j[k]$ denote node $j$'s internal state takes at time instant $k$, $\sigma_{lj}[k]$ denote node $j$'s  broadcasted mass for each out-neighbor $l,~l \in \mathcal{N}_j^+$, and $\tau_{ji}[k]$ denote the total mass received  from node $i \in \mathcal{N}_j^-$. Then, the progress of iteration 1 is described by
% 
%Then, the sender will calculate
%
%In each iteration of the two iteration, at every time step, every node $j$ will calculate three quantities of interest $y_j[k]$ ($z_j[k]$), $\mu_j[k]$ ($\sigma_j[k]$), and $\nu_j[k]$ ($\tau_j[k]$). The quantity $\tau_j[k[$, which is the only quantity that nodes broadcast, is the sum of the $y_j[k]$ calculated by node $j$ up to (and including step $k$) multiplied by the out-going weight of node $j$. 
%
%
%At every time instant $k$, each node $j$ will perform two linear iterations in parallel so as to calculate some $y_j[k]$ and $z_j[k]$ that will be used to calculate the quantity of interest $\pi$. 
%
%
%In this algorithm, at each iteration, the nodes will calculate a weighted linear average of their own $y_j[k[$.
%
%The iteration to calculate $y_j[k]$ performed by node $j$ is given by
\begin{align}
  y_j[k+1] &=\frac{1}{\mathcal{D}_j^+}y_j[k]+\sum_{i \in \mathcal{N}_j^-} \big(\nu_{ji}[k]-\nu_{ji}[k-1]\big), \quad k \geq 0, \nonumber \\
  \mu_{lj}[k] &=\mu_{lj}[k-1]+\frac{1}{\mathcal{D}_j^+}y_j[k]=\sum_{i=0}^{k}\frac{1}{\mathcal{D}_j^+}y_j[i],  \quad k \geq 0, \label{splitting1}  
\end{align}
where 
\begin{align}
\nu_{ji}[k]=\left\{\begin{array}{cc} \mu_{ji}[k], & \text{if } (j,i) \in \mathcal{E}[k],  \quad k \geq 0, \\
\nu_{ji}[k-1], & \text{if } (j,i) \notin \mathcal{E}[k], \quad k \geq 0.\nonumber
\end{array}\right.
\end{align}
Recall that $\mathcal{D}_j^+$ ($\mathcal{D}_i^+$) is the number of nodes that node $j$ ($i$) can transmit information to. Similarly, the progress of iteration 2 is described by
\begin{align}
  z_j[k+1]&=\frac{1}{\mathcal{D}_j^+}z_j[k]+\sum_{i \in \mathcal{N}_j^-}\big(\tau_{ji}[k]-\tau_{ji}[k-1]\big), \quad k \geq 0, \nonumber\\
   \sigma_{lj}[k]&=\sigma_{lj}[k-1]+\frac{1}{\mathcal{D}_j^+}z_j[k]=\sum_{i=0}^{k}\frac{1}{\mathcal{D}_j^+}z_j[i],  \quad k \geq 0,  \label{splitting2}
\end{align}
where 
\begin{align}
\tau_{ji}[k]=\left\{\begin{array}{cc} \sigma_{ji}[k], & \text{if } (j,i) \in \mathcal{E}[k], \quad k \geq 0,\\
\tau_{ji}[k-1], & \text{if } (j,i) \notin \mathcal{E}[k], \quad k \geq 0. \nonumber
\end{array}\right.
\end{align}

As mentioned earlier, for solving the average consensus problem,  the initial conditions in \eqref{splitting1} are set to $y_j[0]=v_j$, whereas the initial conditions in \eqref{splitting2} are set to $z_j[0]=1$. Similarly, for solving the resource coordination problem, the initial conditions in \eqref{splitting1} are set to $y_j[0]=\rho_d/m-\pi_j^{min}$ if $j$ is initially contacted by the leader and $y_j[0]=-\pi_j^{min}$ otherwise, whereas the initial conditions in \eqref{splitting2} are set to $z_j[0]=\pi_j^{max}-\pi_j^{min}>0$. In both the average consensus and coordination problems, $\mu_{ji}[-1]=0$ and $\nu_{ji}[-1]=0$ for all $(j,i) \in \mathcal{E}$, and $\sigma_{ji}[-1]=0$   and $\tau_{ji}[-1]=0$ for all $(j,i) \in \mathcal{E}$.

\textbf{\textit{Main Result:}} We shall argue that with the proposed robustification strategy, despite the presence of unreliable communication links (at each time step, each link  $(j,i) \in \mathcal{E}$, fails independently from other links and independently between time steps, with some probability $1-q_{ji}$), nodes can asymptotically estimate the exact solution $\overline{v}$ to the average consensus by calculating, whenever $z_j[k] >0$  the ratio $y_j[k]/z_j[k]$, i.e., 
\begin{align}
\overline{v}=\lim_{k \rightarrow \infty}\frac{y_j[k]}{z_j[k]}, \label{main_result1}
\end{align}
whenever $z_j[k>0]$. Similarly, exact solution to the resource coordination problem can be obtained as
\begin{align}
& \pi_j=\lim_{k \rightarrow \infty} \big( \pi_j^{min}+ \frac{y_j[k]}{z_j[k]}(\pi_j^{max}-\pi_j^{min})\big),  \label{main_result2} \end{align}
whenever $z_j[k]>0$. In both cases, we run the   iterations in \eqref{splitting1} and \eqref{splitting2} and using the corresponding initial conditions as outlined above. In particular, we will show that, for every $j$, $z_j[k]-  \frac{\sum_{j=1}^n z_0(j)}{\sum_{j=1}^n y_0(j)}  y_j[k] \rightarrow 0$ as $k\rightarrow \infty$ almost surely. Additionally, we will show that $z_j[k]>0$ occurs infinitely often. 

\begin{figure}[b!]
\centering
{\includegraphics[scale=.2]{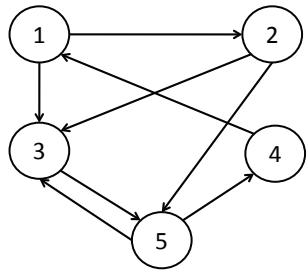}}
\caption{Small directed graph used for illustration of the ratio algorithm for obtaining average consensus in the presence of packet-dropping communication links.}
\label{FIGsmallgraph}
\end{figure}

\subsection{Examples}
We now illustrate how the proposed algorithm works for the case of average consensus in the presence of packet-dropping communication links. We start with the rather small network shown in Fig.~\ref{FIGsmallgraph} and assume that the packets on each link (including the self-links which are not drawn in the figure\footnote{We make this assumption later in the paper for the purposes of simplifying notation.}) can be dropped with probability $1-q$, independently between different links and independently between different iterations. We also assume that the initial values of the five nodes are given by $v=[-4, 5, 6, -3, 1]^T$, with their average equal to $1$. Thus, in the iterations  \eqref{splitting1} and  \eqref{splitting2}
$$
y[0] = [-4, 5, 6, -3, 1]^T, \mbox{ and } z[0] = [ 1, 1, 1, 1, 1]^T \; ,
$$
with $\mu_{ji}[-1] = v_{ji}[-1] = \sigma_{ji}[-1] = \tau_{ji}[-1] = 0$ for all $(j,i) \in \cal{E}$.

We run the iterations in \eqref{splitting1} and \eqref{splitting2}   and plot the ratio $\frac{y_j[k]}{z_j[k]}$ as a function of the iteration step $k$ for each node $j$ ($j = 1, 2, 3, 4, 5$). Figure~\ref{FIGsmallplots} shows the typical behavior that we observe for $q=0.99$ (i.e., for a probability of a packet drop equal to $0.01$). As can be seen in the figure, the ratio at each node quickly converges to the correct average, though the individual values for $y_j[k]$ and $z_j[k]$ do not  converge. 

\begin{figure*}[t!]
    \centering
    \hspace{-0.2in}
    \subfigure[]{
         \includegraphics[width=2.4in]{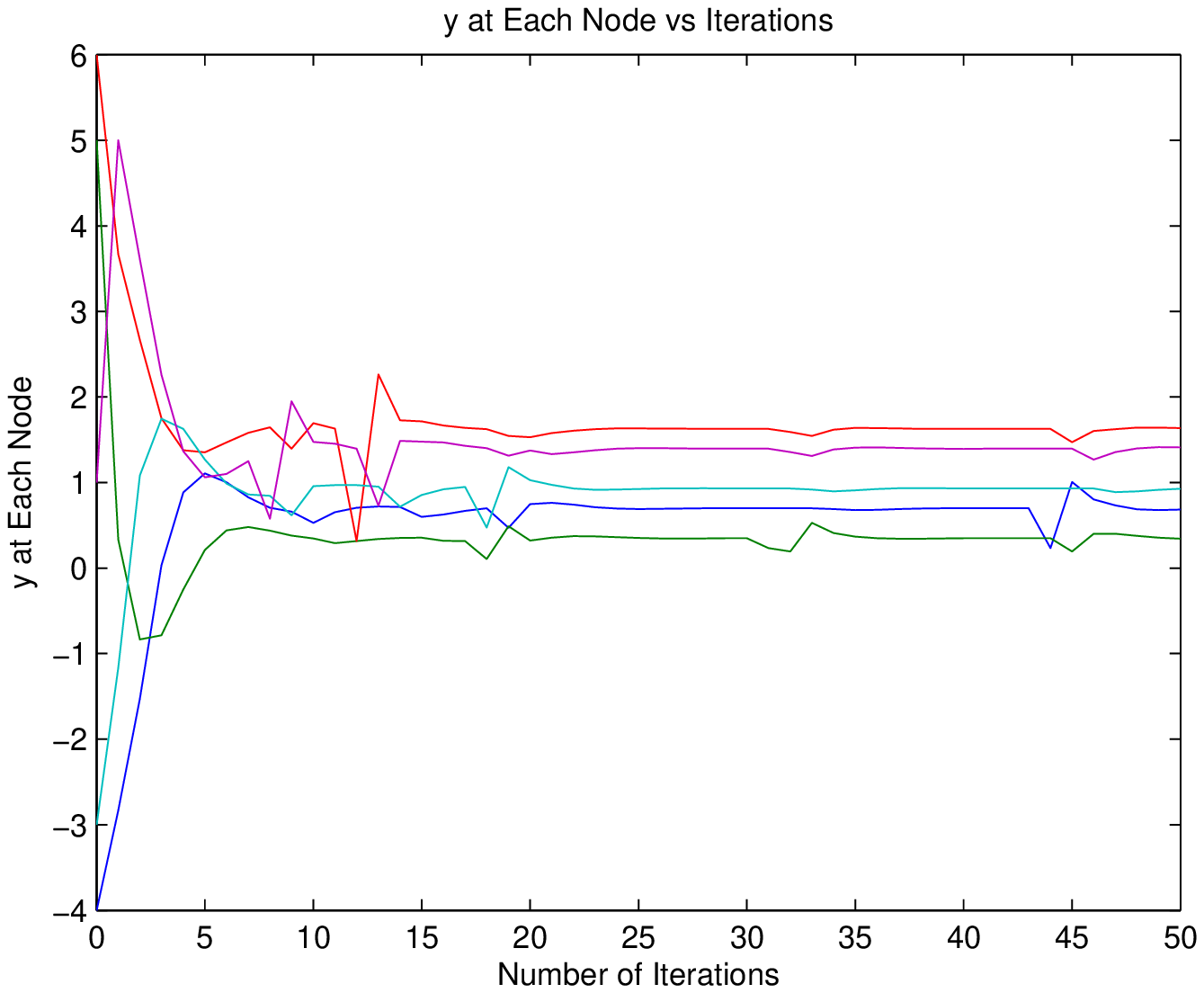}}
             \hspace{-0.2in}
    \subfigure[]{
         \includegraphics[width=2.4in]{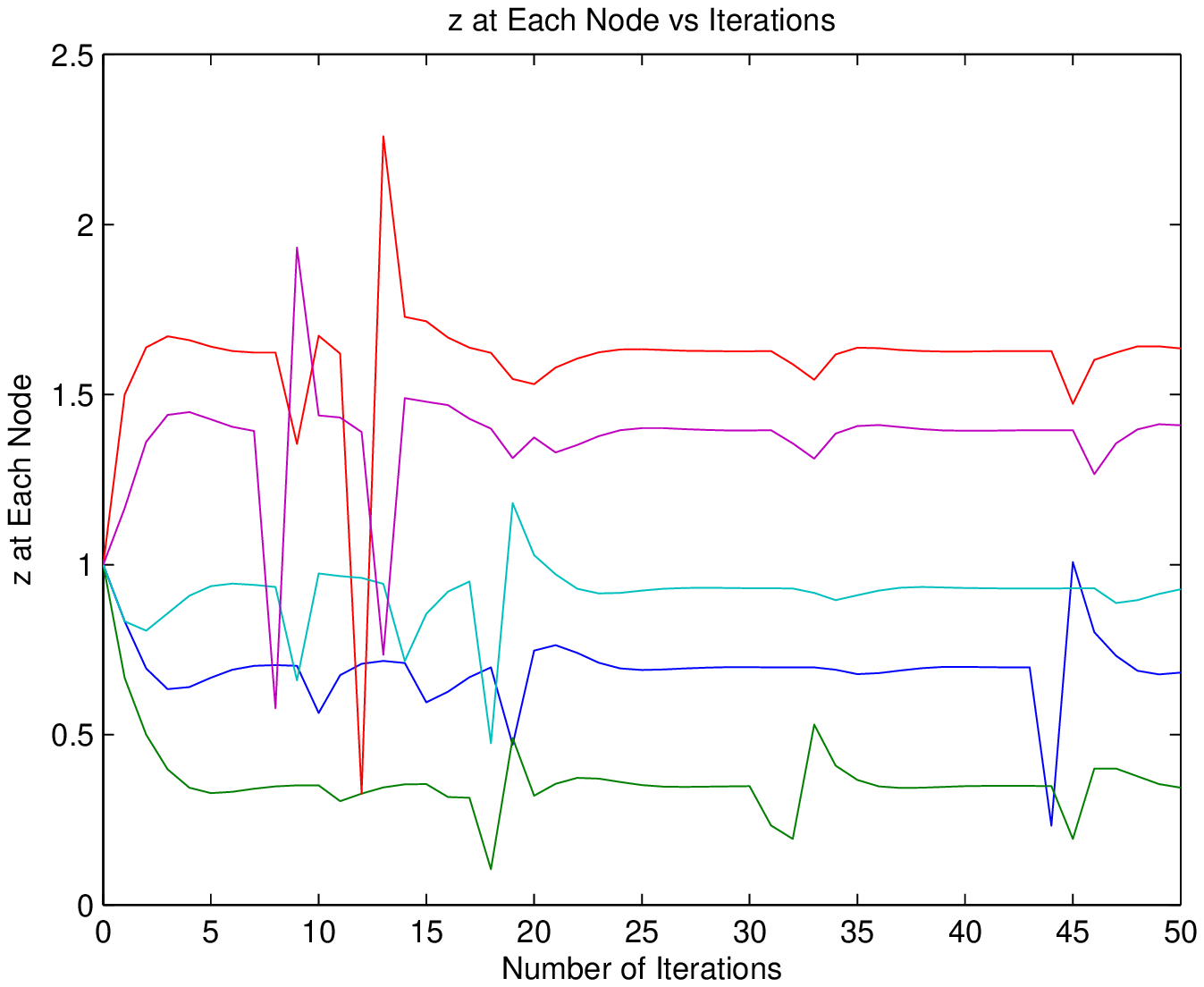}}
             \hspace{-0.2in}
    \subfigure[]{
         \includegraphics[width=2.4in]{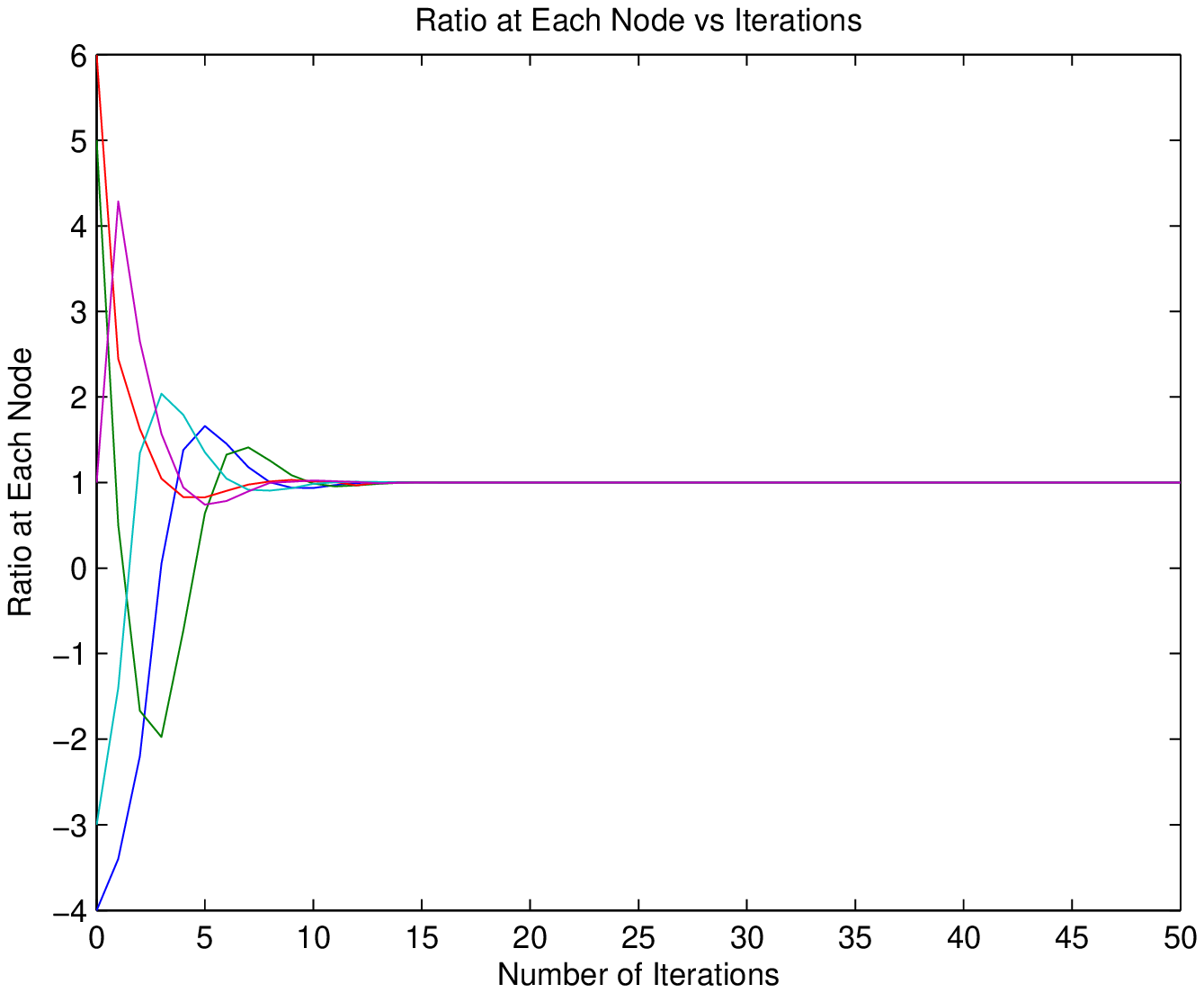}}
             \hspace{-0.2in}
    \caption{Evolution of the values of $y_j[k]$ (left), $z_j[k]$ (middle) and $\frac{y_j[k]}{z_j[k]}$ (right) for $q=0.99$.}
    \label{FIGsmallplots}
\end{figure*}

\begin{figure*}[t!]
    \centering
    \hspace{-0.2in}
    \subfigure[]{
         \includegraphics[width=2.4in]{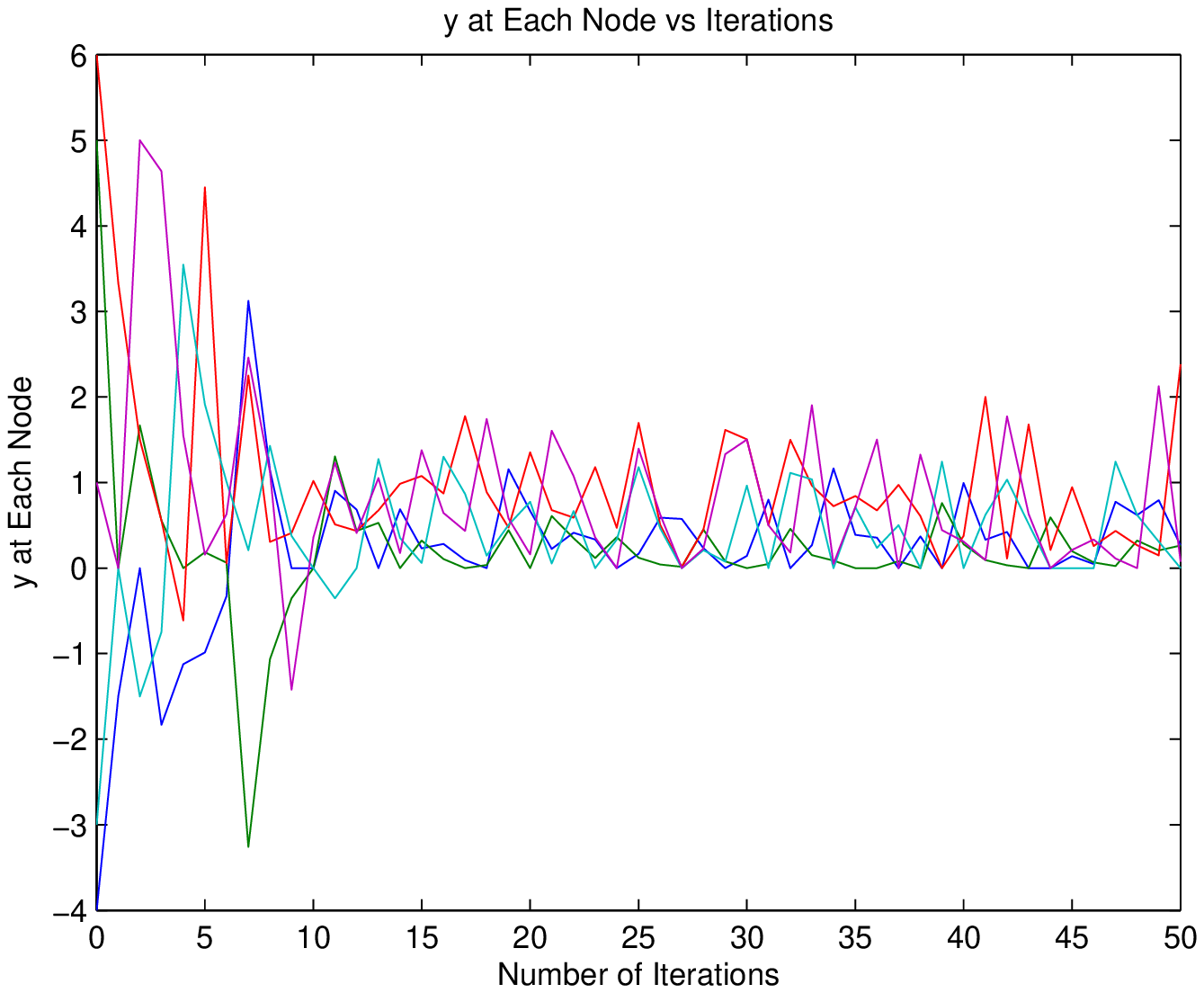}}
             \hspace{-0.2in}
    \subfigure[]{
         \includegraphics[width=2.4in]{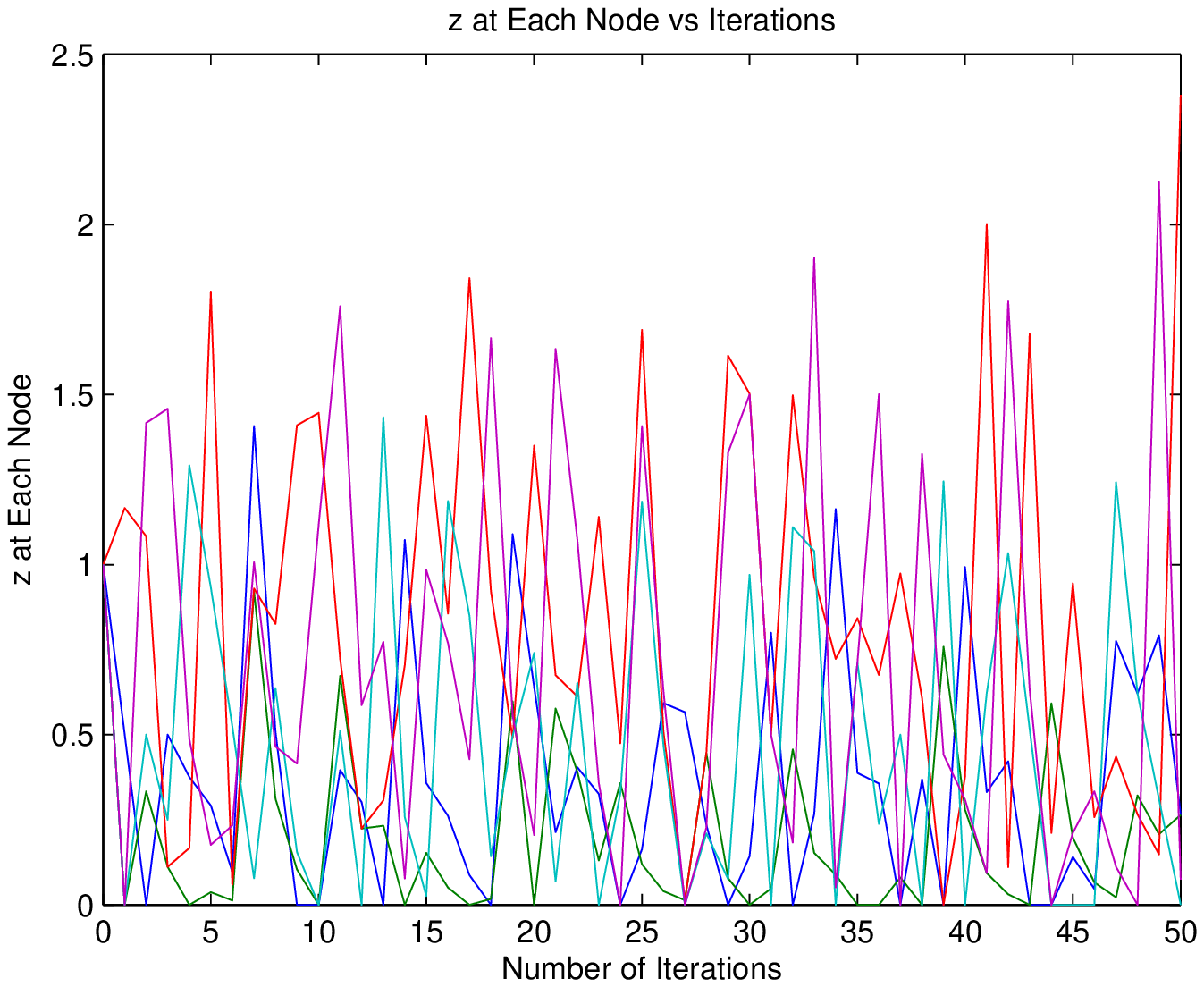}}
             \hspace{-0.2in}
    \subfigure[]{
         \includegraphics[width=2.4in]{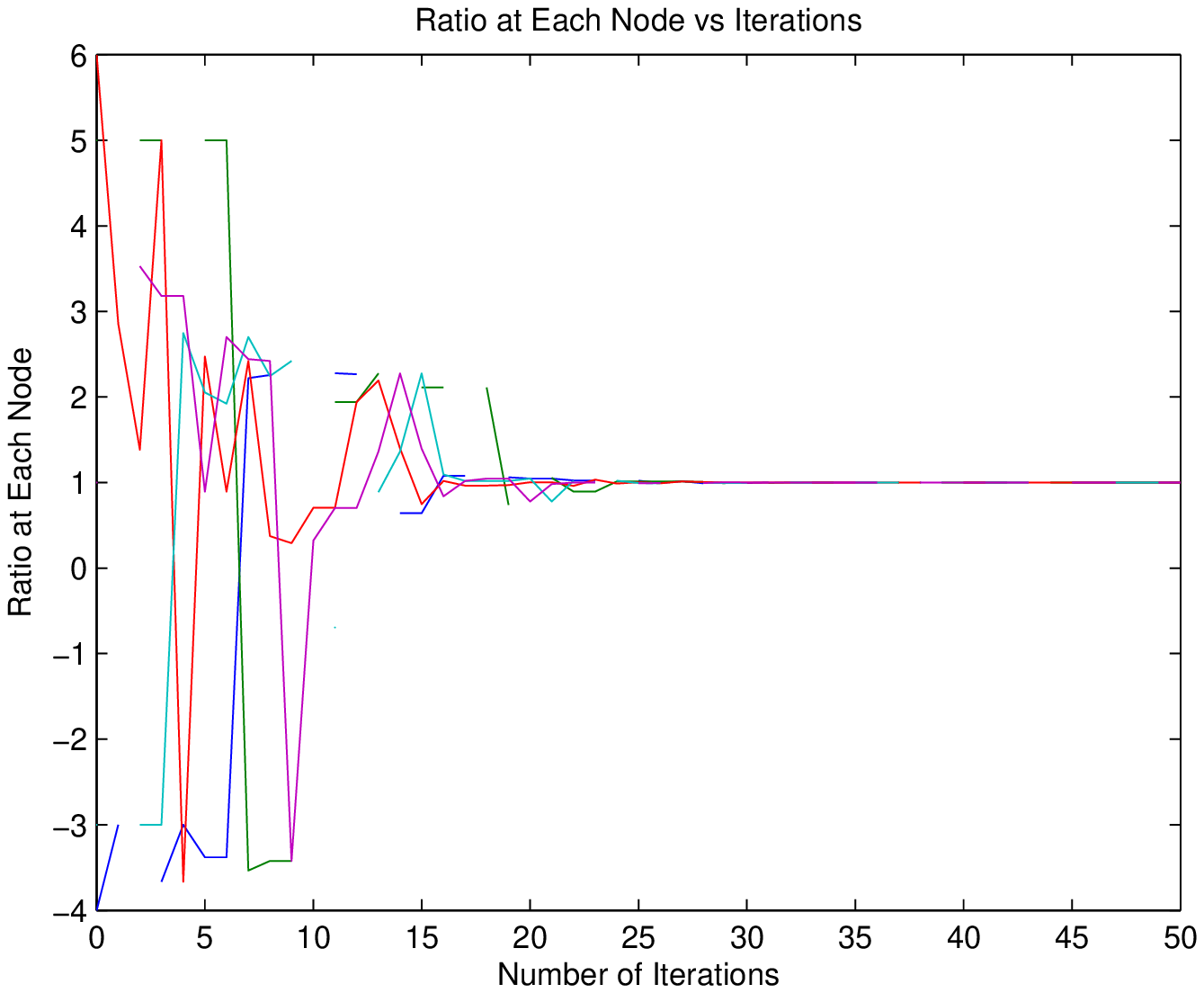}}
             \hspace{-0.2in}
    \caption{Evolution of the values of $y_j[k]$ (left), $z_j[k]$ (middle), and $\frac{y_j[k]}{z_j[k]}$ (right) for $q=0.5$, $j=1,2,3,4,5$.}
    \label{FIGsmallplotsq5}
\end{figure*}

In Figs.~\ref{FIGsmallplotsq5} and~\ref{FIGsmallplotsq9} we show typical behaviors of the same multicomponent system for $q=0.5$ and $q=0.1$. The behavior remains similar to the one observed before: even though $y_j[k]$ and $z_j[k]$ do not converge (in fact, they seem to behave more radically with  decreasing $q$), the ratio $\frac{y_j[k]}{z_j[k]}$ does converge to the average of the initial values. Note that the plot of the ratio in Fig.~\ref{FIGsmallplotsq9} is quite different than the rest: in this case, $q$ is small enough so that $z_j[k]$ (and simultaneously $y_j[k]$) can take the value zero (e.g., when all packets destined to node $j$ are dropped at iteration $k$); thus, the ratio in such cases is not defined and is not plotted, resulting in a discontinuous set of points in the plot. Nevertheless, we can see that when packets are received (which happens frequently enough for each node), the ratio has the correct value. This is a point  addressed  later   in the paper.

\begin{figure*}[t!]
    \centering
    \hspace{-0.2in}
    \subfigure[]{
         \includegraphics[width=2.4in]{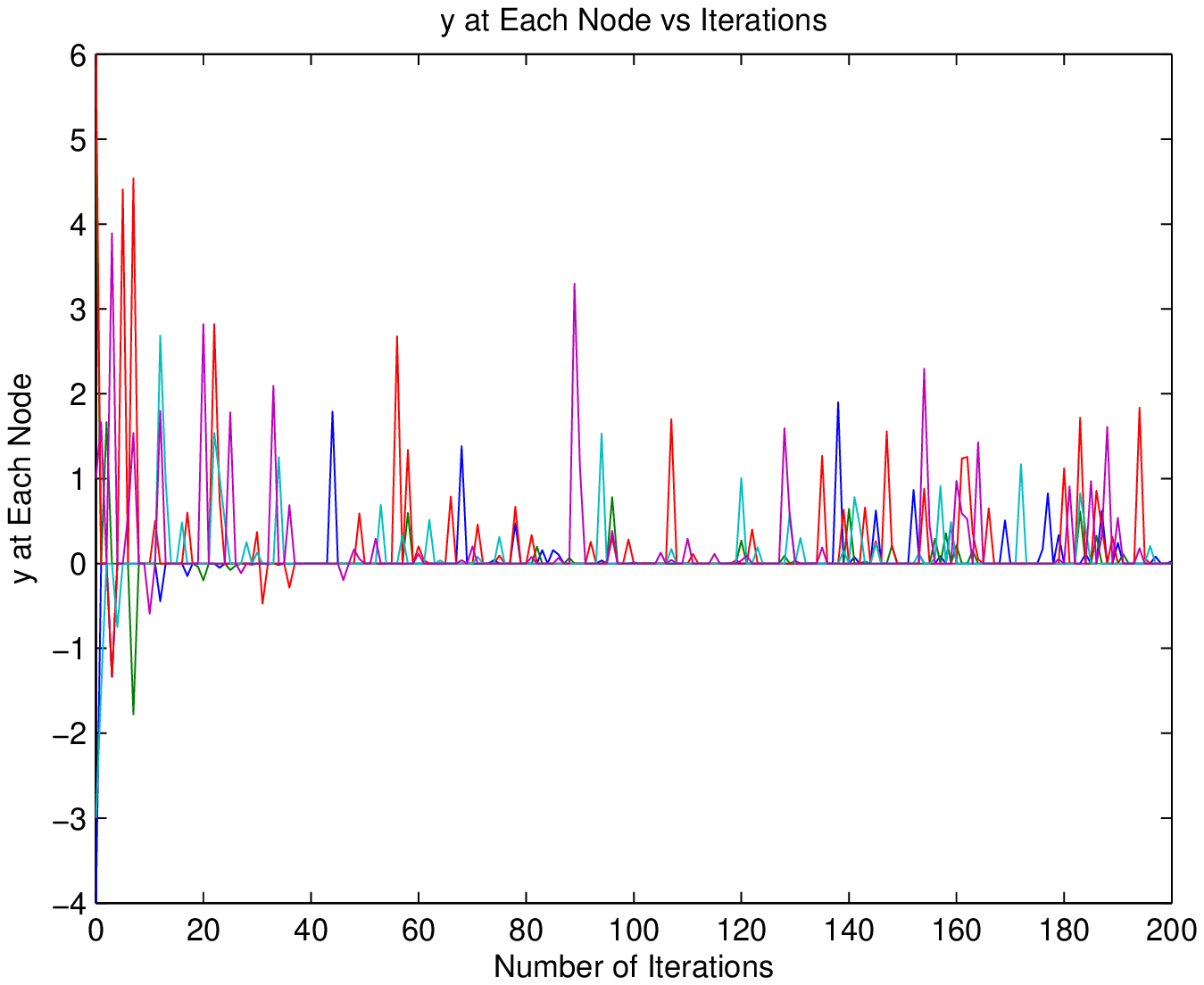}}
             \hspace{-0.2in}
    \subfigure[]{
         \includegraphics[width=2.4in]{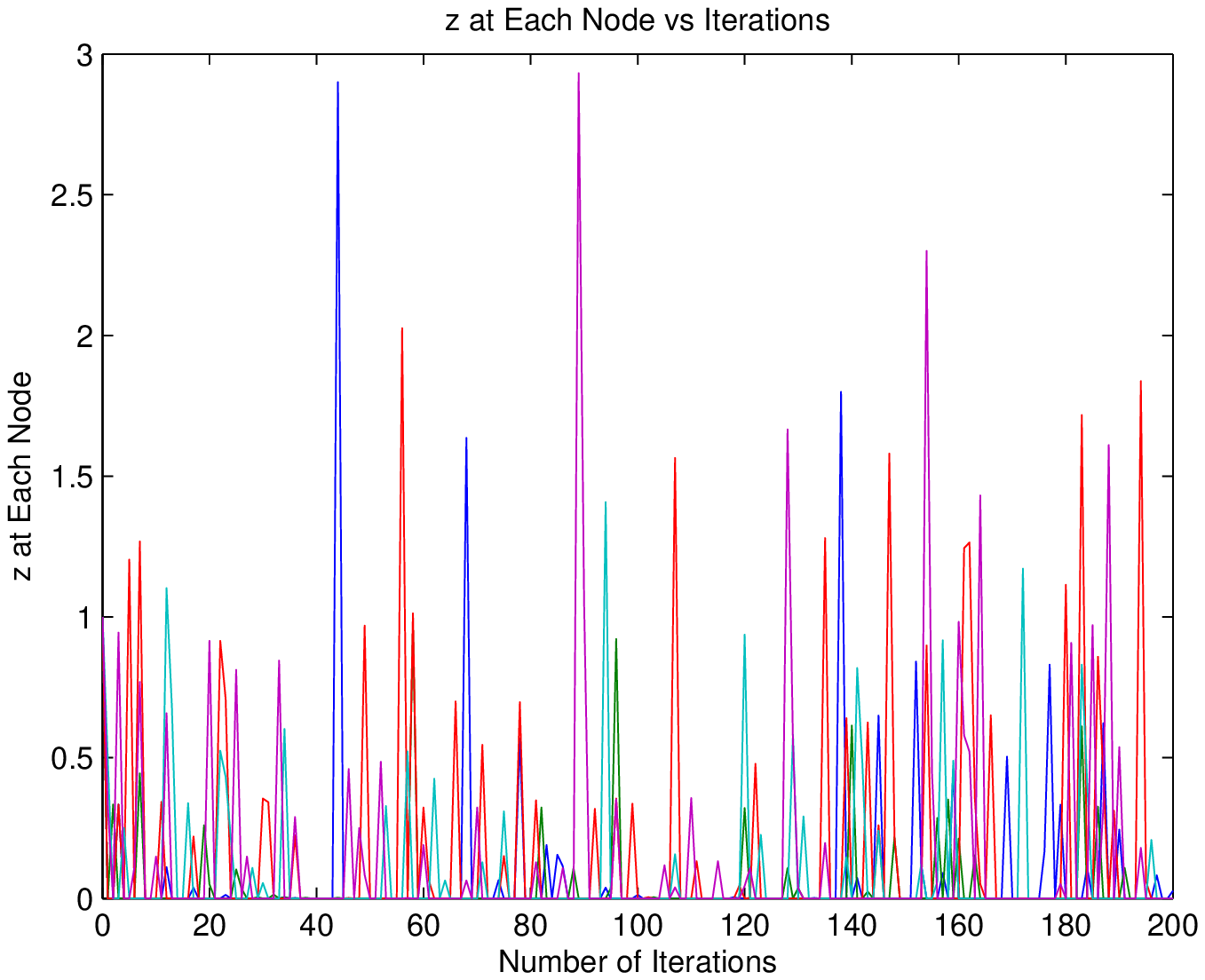}}
             \hspace{-0.2in}
    \subfigure[]{
         \includegraphics[width=2.4in]{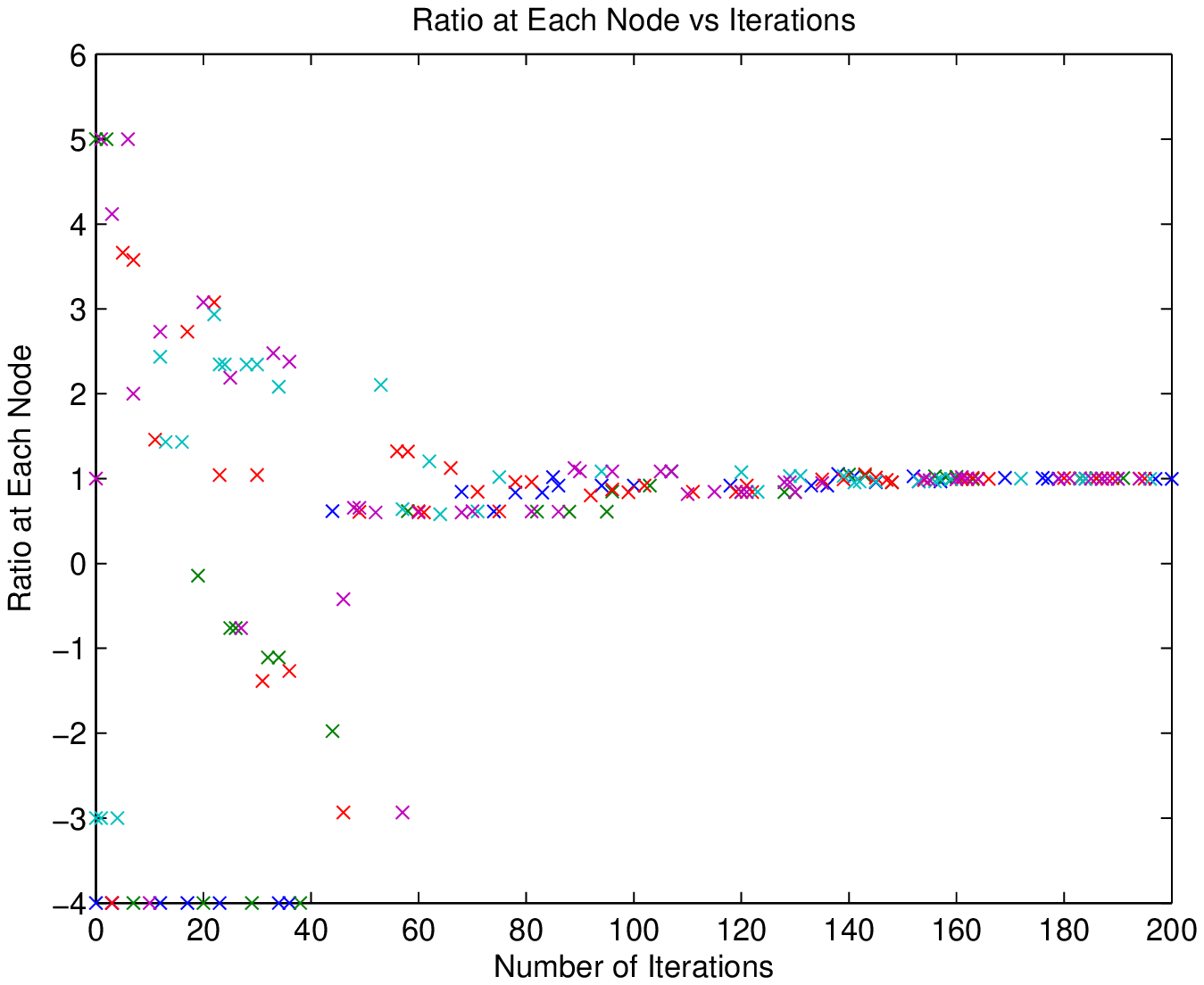}}
             \hspace{-0.2in}
    \caption{Evolution of the values of $y_j[k]$ (left), $z_j[k]$ (middle), and $\frac{y_j[k]}{z_j[k]}$ (right) for $q=0.1$, $j=1,2,3,4,5$.}
    \label{FIGsmallplotsq9}
\end{figure*}

\begin{figure*}[t!]
    \centering
    \hspace{-0.2in}
    \subfigure[]{
         \includegraphics[width=2.4in]{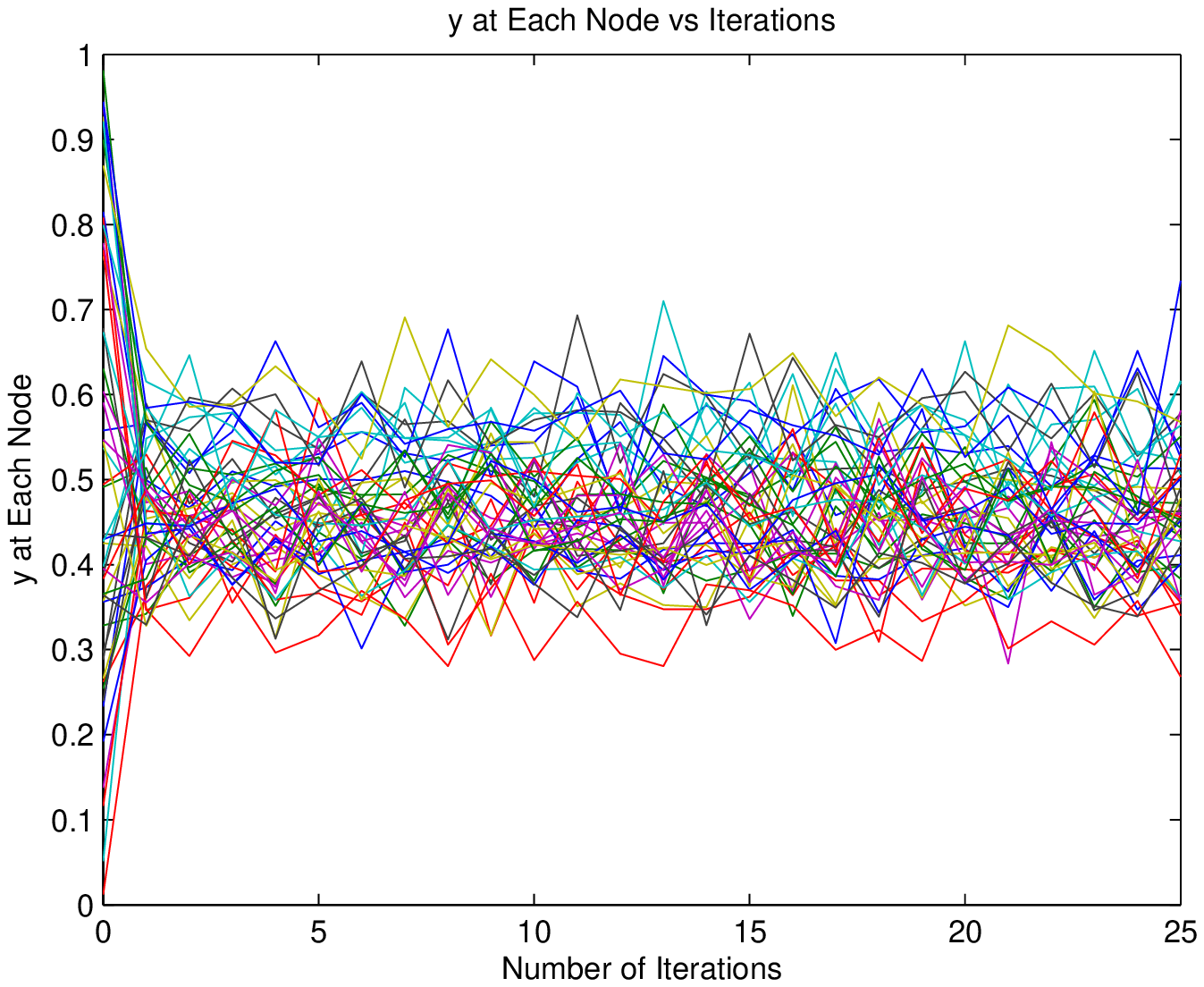}}
             \hspace{-0.2in}
    \subfigure[]{
         \includegraphics[width=2.4in]{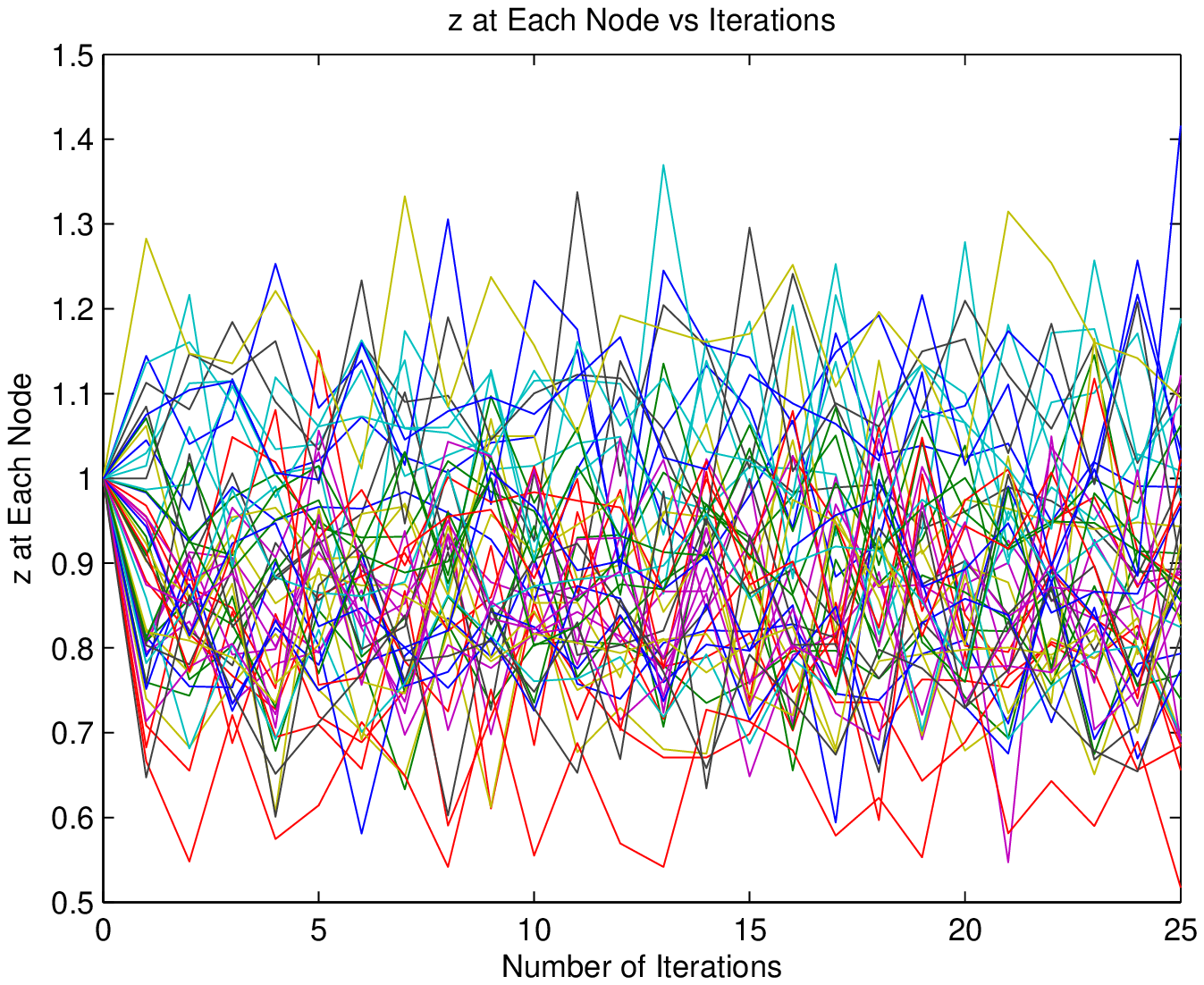}}
             \hspace{-0.2in}
    \subfigure[]{
         \includegraphics[width=2.4in]{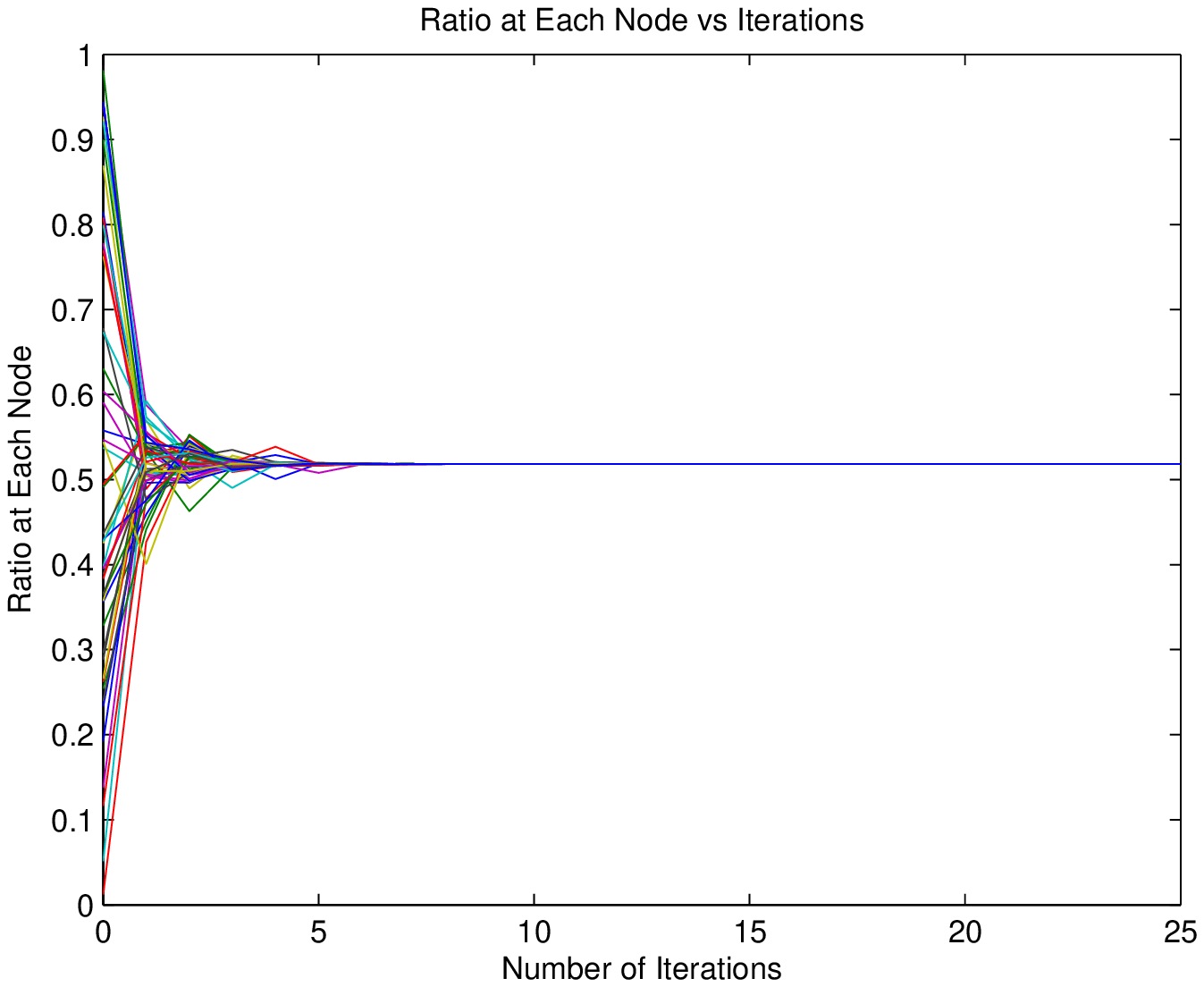}}
             \hspace{-0.2in}
    \caption{Evolution of the values of $y_j[k]$ (left), $z_j[k]$ (middle), and $\frac{y_j[k]}{z_j[k]}$ (right) for $q=0.1$ for a 50-component system.}
    \label{FIGlargeplotsp1}
\end{figure*}

An example of what happens in larger graphs is shown in Fig.~\ref{FIGlargeplotsp1}. Here we consider a graph with $50$ nodes, randomly generated by choosing a directed edge from node $i$ to node $j$, $1 \leq i,j \leq 50$, $i \neq j$, independently with probability $1/2$, and ensuring that the resulting graph is strongly connected. As can be seen the behavior remains similar to what we observed for the smaller graph: the ratio $\frac{y_j[k]}{z_j[k]}$ converges quickly to the average even though the individual $y_j[k]$ and $z_j[k]$ do not converge. For this particular plot, we used $q=0.1$, which also justifies the fluctuation in the values of $y_j[k]$ and $z_j[k]$.

% and $\mu_j[-1]=\nu_j[-1]=0$.
%
%For solving the average consensus problem,  the initial conditions in \eqref{splitting1} are set to $y_j[0]=v_j$ and $\mu_j[-1]=\nu_j[-1]=0$
%
%
% The initial conditions in \eqref{splitting1} are set to $y_j[0]=\pi_j^{max}-\pi_j^{min} = \ell_j>0$ and $\mu_j[-1]=\nu_j[-1]=0$. 
%
%
%Then, the initial conditions in \eqref{splitting2} are set to $z_j[0]=\rho_d/m-\pi_j^{min}$ if $j$ is a neighbor of the leader and $z_j[0]=-\pi_j^{min}$ otherwise, and $\sigma_j[-1]=\tau_j[-1]=0$. Additionally, at every $k$, each node computes
%\begin{align}
%& \pi_j[k]=\pi_j^{min}+ \frac{z_j[k]}{y_j[k]}\ell_j, \label{splitting3}
%\end{align}
%which asymptotically converges to a quantity $\pi_j$ given by
%\begin{align}
%& \pi_j= \lim_{k \rightarrow \infty}\pi_j[k]=\pi_j^{min}+  \frac{\rho_d-\sum_{l=1}^n \pi_l^{min}}{\sum_{l=1}^n\ell_l} \ell_j.
%\end{align}
%where $\ell_j =\pi_j^{max}-\pi_j^{min},~\forall j$.

\section{First and Second Moment Analysis}
In this section, we obtain recurrence relations that describe the first and second moment of the iterations after \eqref{splitting1} and \eqref{splitting2}; this analysis is used in Section~\ref{convergence_analysis} to establish the claims in \eqref{main_result1} and \eqref{main_result2}. In order to ease the moment calculations, the expressions in  \eqref{splitting1}--\eqref{splitting2} will be rewritten more compactly in vector form. Also,  in order to facilitate notation, we will allow each node $j$ to drop the packet carrying its own previous value when updating its  value. This way, node $j$ handles its own value in the same way as its neighbors' values and notation is simplified significantly. 

\subsection{Vectorized Description of Double-Iteration Algorithm} 
Using the definition for the indicator variable $x_{ji}[k]$ given in \eqref{Bernouilli_model} and the resulting indicator variable $\ell_{ji}[k]$ given in \eqref{link_ji}, which describes the successful transmission of information from node $i$ to node $j$ over an existing, unreliable communication link, iterations \eqref{splitting1}  and \eqref{splitting2} can be rewritten as
\begin{align}
 \mu_{lj}[k]&=\left\{\begin{array}{cc} \mu_{lj}[k-1]+\frac{1}{\mathcal{D}_j^+}y_j[k], & \text{if } l \in \mathcal{N}_j^+, \quad k\geq 0, \\
0, & \text{if } l \notin \mathcal{N}_j^+,  \quad k\geq 0, \label{iter_1_gen_1}   
\end{array}\right . \\
\nu_{ji}[k]&=\left\{\begin{array}{cc} \mu_{ji}[k]x_{ji}[k]+\nu_{ji}[k-1](1-x_{ji}[k]), &  \text{if } i \in \mathcal{N}_j^-, \quad k\geq 0, \\
0, & \text{if } i \notin \mathcal{N}_j^-,  \quad k\geq 0, \label{iter_1_gen_2} 
\end{array}\right. \\
y_j[k+1] &=\sum_{i =1}^n \big(\nu_{ji}[k]-\nu_{ji}[k-1]\big), \quad k\geq 0, \label{iter_1_gen_3}  
\end{align}
and
\begin{align}
 \sigma_{lj}[k]&=\left\{\begin{array}{cc} \sigma_{lj}[k-1]+\frac{1}{\mathcal{D}_j^+}z_j[k], & \text{if } l \in \mathcal{N}_j^+,  \quad k\geq 0, \\
0, & \text{if } l \notin \mathcal{N}_j^+,   \quad k\geq 0,  \label{iter_2_gen_1}   
\end{array}\right. \\
\tau_{ji}[k]&=\left\{\begin{array}{cc} \sigma_{ji}[k]x_{ji}[k]+\tau_{ji}[k-1](1-x_{ji}[k]), &  \text{if }  i \in \mathcal{N}_j^-,  \quad k\geq 0, \\
0, & \text{if }  i \in \mathcal{N}_j^-,  \quad k\geq 0,  \label{iter_2_gen_2} 
\end{array}\right. \\
z_j[k+1] &=\sum_{i =1}^n \big(\tau_{ji}[k]-\tau_{ji}[k-1]\big), \quad k\geq 0,  \label{iter_2_gen_3}  
\end{align}
where $\mu_{lj}[-1]= \nu_{ji}[-1]=  \sigma_{lj}[-1]= \tau_{ji}[-1]=0,~\forall j,i$.

Let $A\circ B$ denote the  Hadamard (entry-wise) product of a pair of matrices $A$ and $B$ of identical size. Then, for all $k\geq 0$, iteration  \eqref{iter_1_gen_1}--\eqref{iter_1_gen_3}   can be rewritten in matrix form as
\begin{align}
& M_{k}=M_{k-1}+P\text{diag}(y_k),\label{mu_dynamics_1}\\
& N_k=M_k \circ X_k+N_{k-1} \circ (U-X_k), \label{mu_dynamics_2}\\
& y_{k+1}=(N_{k}-N_{k-1}) e=\big[(M_{k}-N_{k-1}) \circ X_{k}\big] e,   \label{mu_dynamics_3}
\end{align}
where  $P=[p_{ji}] \in \mathbb{R}^{n \times n}$, with $p_{ji}=\frac{1}{\mathcal{D}_i^+},~\forall j \in\mathcal{N}_i^+$ and $p_{ji}=0$ otherwise; $M_{-1}=N_{-1}=0$; $y_k=y[k]$; $U \in \mathbb{R}^{n \times n}$, with $[U_{ji}]=1,~\forall i,j$; $\text{diag}(y_k)$ is the  diagonal matrix that results by having the entries of $y_k$ on the main diagonal; and $e=[1,1,\dots,1]'$ (note that $U=ee^{T}$). Similarly, for $k \geq 0$, \eqref{iter_2_gen_1}--\eqref{iter_2_gen_3}   can be rewritten in matrix form as
\begin{align}
& S_{k}=S_{k-1}+P\text{diag}(z_k),  \label{mu_dynamics_4}\\
& T_k=S_k \circ X_k+T_{k-1} \circ (U-X_k),  \label{mu_dynamics_5}\\
& z_{k+1}=(T_{k}-T_{k-1}) e=\big[(S_{k}-T_{k-1}) \circ X_{k}\big] e,  \label{mu_dynamics_6}
\end{align}
where $S_{-1}=T_{-1}=0$,  $z_k=z[k]$, and $\text{diag}(z_k)$ is the diagonal matrix that results by having  the entries of $z_k$ on the main diagonal.

By defining $A_{k}:=M_{k}-N_{k-1}$ and $B_{k}:=S_{k}-T_{k-1}$, iteration \eqref{mu_dynamics_1}--\eqref{mu_dynamics_3}  can be rewritten more compactly as  
\begin{align}
& A_{k} = A_{k-1} \circ (U-X_{k-1})+P\text{diag}(y_k), \quad k\geq 1, \label{y_dynamics_a} \\
&  y_{k+1}=(A_{k} \circ X_{k})e,\quad k\geq 0,   \label{y_dynamics_b}
\end{align}
and iteration \eqref{mu_dynamics_4}--\eqref{mu_dynamics_6} as 
\begin{align}
& B_{k} = B_{k-1} \circ (U-X_{k-1})+P\text{diag}(z_k), \quad k\geq 1,\label{z_dynamics_a} \\
&  z_{k+1}=(B_{k} \circ X_{k})e,  \quad k\geq 0, \label{z_dynamics_b}
\end{align}
where $A_0=M_0-N_{-1}=P\text{diag}(y_0)$, and $B_0=S_0-T_{-1}=P\text{diag}(z_0)$.

For analysis purposes, each matrix in \eqref{y_dynamics_a}--\eqref{y_dynamics_b} and \eqref{z_dynamics_a}--\eqref{z_dynamics_b} will be rewritten in vector form by stacking up the corresponding columns.\footnote{If we let $A=[A_{ij}] \in \mathbb{R}^{n \times n}$, then $a=[A_{11},A_{21},\dots,A_{n1},A_{12},A_{22},\dots,A_{n2},\dots,A_{1n},A_{2n},\dots,A_{nn}]^T$. Vectors defined by stacking the columns of a matrix will be denoted with the same small letter as the capital letter of  the corresponding matrix.}  Then, \eqref{y_dynamics_a}--\eqref{y_dynamics_b} and \eqref{z_dynamics_a}--\eqref{z_dynamics_b} can be rewritten in vector form as follows. Let $F=[I_n~ I_n~ \dots~ I_n] \in \mathbb{R}^{n \times n^2}$, where $I_n$ is the $n \times n$ identity matrix, and $\tilde{P}=[E_1P^T~E_2P^T~\dots~E_nP^T ]^T \in \mathbb{R}^{n^2 \times n}$,  where $E_i \in \mathbb{R}^{n \times n}$ has $E_i(i,i)=1$ and all other entries equal  zero. [The entries of $E_iP^T   \in \mathbb{R}^{n \times n}$  ($PE_i^T=PE_i$) are all zero except for the  $i^{th}$ row (column) entries, which are those of the $i^{th}$ row (column) of matrix $P^T$ ($P$).] Then,  \eqref{y_dynamics_a}--\eqref{y_dynamics_b} can be rewritten as
\begin{align}
& a_{k}=a_{k-1} \circ (u-x_{k-1}) +\tilde P y_k, \quad k \geq 1, \label{ay_dynamics_1} \\
& y_{k+1}=F(a_k \circ x_k),  \quad k \geq 0, \label{ay_dynamics_2}
\end{align}
where $a_{k} \in \mathbb{R}^{n^2}$, $x_k  \in \mathbb{R}^{n^2}$, and $x_{k-1}  \in \mathbb{R}^{n^2}$  result from stacking the columns of matrices $A_{k}$, $X_{k}$, and $X_{k-1}$, respectively. Similarly,  \eqref{z_dynamics_a}--\eqref{z_dynamics_b} can be rewritten as
\begin{align}
& b_{k}=b_{k-1} \circ (u-x_{k-1})+\tilde P z_k, \quad k \geq 1, \label{by_dynamics_1} \\
& z_{k+1}=F(b_k \circ x_k), \quad k \geq 0, \label{by_dynamics_2}
\end{align}
where $b_{k} \in \mathbb{R}^{n^2}$  results from stacking the columns of matrix $B_{k}$.

\begin{remark}
It is important to note that matrices $A_k$ and $B_k$, and their corresponding vectors $a_k$ and $b_k$, have some entries that remain at zero for all $k \geq 0$. Specifically, the $(j,i)$ entry of matrices $A_k$ and $B_k$ (and their corresponding entries in $a_k$ and $b_k$)  remain  zero if  there is no communication link from node  $i$ to node  $j$, i.e., $(j,i) \notin \mathcal{E}$. The reason we keep these entries (despite the fact  they are zero and do not play a role in the analysis) is because they facilitate matrix notation and calculations in subsequent developments. \qed
\end{remark}

Since it will appear later at several points of the analysis, it is worth noting that when premultiplying $\tilde{P}$ by $F$, we recover the matrix $P$, i.e., 
\begin{align}
P=F\tilde{P}. \label{FtildeP}
\end{align}

\subsection{First Moment Analysis}
%%% CHRIS
\label{SUBfirstmoment}
In this section, we describe the first moment dynamics of \eqref{ay_dynamics_1}--\eqref{by_dynamics_2} via discrete-time transition systems with no inputs, where (as shown below) the corresponding transition matrices (which  are obtained from $P$ and $q$) are column stochastic and primitive. In both iterations, the sum of the entries of the first moment vectors for $y_k$ and $z_k$  is shown to remain constant over time and be respectively equal to the sum of $q\sum_iy_0(i)$ and  $q\sum_iz_0(i)$. Furthermore, both first moments $\Ex[y_k]$ and $\Ex[z_k]$ are shown to reach a steady-state value as $k$ goes to infinity. The above discussion is formalized in the following lemma.

\begin{lemma} \label{first_moment} 
Let $a_k$, $b_k$, $y_k$, and $z_k$ be described by the recurrence relations in \eqref{ay_dynamics_1}--\eqref{ay_dynamics_2}, and  \eqref{by_dynamics_1}--\eqref{by_dynamics_2} respectively. Let the first moments of   $a_k$, $y_k$, $b_k$, and $z_k$ (i.e., $\Ex[a_k]$, $\Ex[y_k]$, $\Ex[b_k]$, and $\Ex[z_k]$) be denoted by $\overline{a}_k$, $\overline{y}_k$, $\overline{b}_k$, and $\overline{z}_k$ respectively. Then the evolution of $\overline{a}_{k}$,  $\overline{y}_{k}$, $\overline{b}_{k}$, and $\overline{z}_{k}$, $\forall k\geq 1$, is governed by
\begin{align}
& \overline{a}_{k}=\big[q\tilde{P}F+(1-q)I_{n^2} \big]\overline{a}_{k-1}, \\
& \overline{y}_{k+1}=\big[qP+(1-q)I_{n } \big]\overline{y}_k, \\
& \overline{b}_{k}=\big[q\tilde{P}F+(1-q)I_{n^2} \big]\overline{b}_{k-1}, \\
& \overline{z}_{k+1}=\big[qP+(1-q)I_{n } \big]\overline{z}_k,
\end{align}
where $I_{m}$ is the $m \times m$ identity matrix, with $\overline{a}_0=\tilde{P}y_0$, $\overline{y}_1=qPy_0$, $\overline{b}_0=\tilde{P}z_0$, and $\overline{z}_1=qPz_0$.
\end{lemma}

\begin{proof}
Since the development  for obtaining $\overline{a}_k$ and $\overline{y}_k$ is parallel to that for obtaining $\overline{b}_k$ and $\overline{z}_k$, our analysis focuses on the first case. For $k=0$ in \eqref{ay_dynamics_1}--\eqref{ay_dynamics_2}, by taking expectations of both sides and noting that packet drops at time step $k=0$ are independent of the initial values for $a_0$, it follows that
\begin{align}
& \overline{a}_{0}=\tilde{P}y_0, \label{mu_dynamics_1a}\\
& \overline{y}_1=qF\overline{a}_{0}. \label{mu_dynamics_3a}
\end{align}
Substituting \eqref{mu_dynamics_1a} into \eqref{mu_dynamics_3a}, we obtain $\overline{y}_1=qF\tilde{P}\overline{y}_0=qP\overline{y}_0$. 

For $k\geq 1$ in   \eqref{ay_dynamics_1}--\eqref{ay_dynamics_2}, noting that packet drops at time step $k$ are independent of previous packet drops and the initial values of $a_0$,  it follows, by taking expectations on both sides, that
\begin{align}
& \overline{a}_{k}=\overline{a_{k-1} \circ (u-x_{k-1})}+\overline{\tilde{P}y_k}=\overline{a}_{k-1} \circ (u-\overline{x}_{k-1})+ \tilde{P} \overline{y}_k=(1-q)\overline{a}_{k-1}+\tilde{P}\overline{y}_{k}, \label{mu_avg_dyn_1} \\
&\overline{y}_{k+1} =\overline{F(a_{k} \circ x_{k})}= F (\overline{a}_{k}  \circ \overline{x}_{k} )=qF\overline{a}_{k}. \label{mu_avg_dyn_2}
\end{align}
Substituting  \eqref{mu_avg_dyn_2} into  \eqref{mu_avg_dyn_1}, we obtain
\begin{align}
\overline{a}_{k}&=(1-q)\overline{a}_{k-1}+q\tilde{P}F\overline{a}_{k-1}  \\
& =[q\tilde{P}F+(1-q)I_{n^2}]\overline{a}_{k-1},
\end{align}
Similarly, substituting \eqref{mu_avg_dyn_1} into \eqref{mu_avg_dyn_2}, we have
\begin{align}
\overline{y}_{k+1}&=(1-q)qF\overline{a}_{k-1} +qF\tilde{P}\overline{y}_k  \\
& =(1-q)\overline{y}_k+qF\tilde{P}\overline{y}_k  \\
& =[qP+(1-q)I_n]\overline{y}_k,
\end{align}
where $I_n$ is the $n \times n$ identity matrix.
\end{proof}

\subsection{Second Moment Analysis}
%%% CHRIS
\label{SUBsecondmoment}
In the order to calculate the second moment  dynamics for  \eqref{ay_dynamics_1}--\eqref{by_dynamics_2}, we utilize in the following lemma.

\begin{lemma} \label{exp_lemma}
Let $x$, $c$ and $d$ be random vectors of  dimension $n$. Furthermore, assume that the entries of $x$ are Bernoulli \textit{i.i.d.} random variables such that $\Pr \{ x_{i}=1\}=q $ and $\Pr \{x_{i}=0\}=1-q $, $\forall i=1,2,\dots n$, and are independent from $c$ and $d$. Then
\begin{align}
& S:=\Ex  \big [ (c  \circ x) (x  \circ d )^T ]=q^2 \Ex  [c d^T]+q(1-q)\Ex \big[\text{diag}(cd^T)], \label{exp_lemma_1} \\
& T:=\Ex  \big [ \big(c \circ x\big)\big((u-x)  \circ d\big)^T ]=q(1-q) \Ex  [cd^T]-q(1-q)\Ex \big[\text{diag}(cd^T)], \label{exp_lemma_2}
\end{align}
where $\text{diag}(cd^T)$ is a diagonal matrix with the same diagonal as matrix $cd^T$.
\end{lemma}
\begin{proof}
The $(i,j), i\neq j$, entry of $S$ can be obtained as follows:
\begin{align}
S_{ij}=\Ex \big[  c_{i }x_{i}d_{j }x_{j} \big].
\end{align}
Since $x_{i}$ and $x_{j}$ are pairwise independent, and independent from $c$ and $d$, it follows that
\begin{align}
\Ex \big [  c_{i  }x_{i }d_{j  }x_{j } \big]=q^2  \Ex \big [c_{i }d_{j }\big]  . \label{exp_1}
\end{align}
For $i=j$,  observing that $\Ex[x_{i}x_{i}]=\Ex[x_{i }]=q, ~\forall i=1, \dots, n$, we obtain the corresponding entry of $S$  as 
\begin{align}
S_{ii}&=\Ex \big  [c_{i}x_{i}d_{i  }x_{i } \big]=\Ex \big [ c_{i   }d_{i  }x_{i }\big] =q  \Ex \big [c_{i }d_{i  } \big]\label{exp_2}.
\end{align}
In \eqref{exp_1}, it is easy to see that $  \Ex \big [c_{i }d_{j }\big]$ is the $(i,j)$ entry of $\Ex  [cd^T]$. Similarly, in \eqref{exp_2}, it is easy to see that $ \Ex \big [c_{i }d_{i  } \big]$ is the $(i,i)$  entry of $\Ex  [cd^T ]$. From these observations,  the result in \eqref{exp_lemma_1} follows.

Similarly, the $(i,j), i\neq j$, entry of $T$ can be obtained as follows:
\begin{align}
T_{ij}=\Ex \big[ c_{i  }x_{i }d_{j  }(1-x_{j }) \big].
\end{align}
Since $x_{i }$ and $(1-x_{j })$ are independent, it follows that
\begin{align}
\Ex \big [ c_{i }x_{i }d_{j  }(1-x_{j }) \big]=  \Ex \big [c_{i  }d_{j   }\big]\Ex \big [x_{i   }(1-x_{j   })\big]  =q(1-q) \Ex \big [c_{i }d_{j }\big]    . \label{exp_4}
\end{align}
For $i=j$, and observing that $\Ex[x_{i}(1-x_{i })]=0, ~\forall i=1, \dots, n$, the corresponding entry of $T$ can obtained as follows;
\begin{align}
T_{ii}&=\Ex \big[ c_{i  }x_{i }d_{i   }(1-x_{i}) \big]= \Ex \big [c_{i   }d_{i}\big]\Ex \big [x_{i   }(1-x_{i   })\big]  =0 \label{exp_5}.
\end{align}
The result in \eqref{exp_lemma_2} follows from \eqref{exp_4} and \eqref{exp_5}. 
\end{proof}

%For $aa^T$, $\Gamma$ $\Phi$
%
%For $bb^T$, $\Psi$, $\Lambda$
%
%For $ab^T$, $\Xi$, $\Upsilon$  

The following lemma establishes that the evolution of $\Ex[a_ka_k^T]$, $\Ex[b_kb_k^T]$, and $\Ex[a_kb_k^T]$, and  can be expressed as linear iterations with identical dynamics but different initial conditions. Similarly, the evolution of  $\Ex[y_ky_k^T]$, $\Ex[z_kz_k^T]$, and  $\Ex[y_kz_k^T]$ can also be expressed as linear iterations with identical dynamics but different initial conditions. 
\begin{lemma} \label{second_moment}
Consider the second moments of $a_k$, $y_k$, $b_k$, and $z_k$, and let $\Ex[a_ka_k^T]$, $\Ex[y_ky_k^T]$, $\Ex[b_kb_k^T]$, $\Ex[z_kz_k^T]$, $\Ex[a_kb_k^T]$, and $\Ex[y_kz_k^T]$) be denoted by $\Gamma_k$, $\Phi_k$, $\Psi_k$, $\Lambda_k$, $\Xi_k$, and $\Upsilon_k$ respectively. Then, the evolutions of $\Gamma_k$, $\Phi_k$, $\Psi_k$, $\Lambda_k$, $\Xi_k$, $\Upsilon_k,~\forall k\geq 1$, are described by the following iterations (where all $I$ denote $n^2 \times n^2$ identity matrices):
\begin{align}
 \Gamma_{k}&=\big[q\tilde{P}F+(1-q)I \big]\Gamma_{k-1}\big[q\tilde{P}F+(1-q)I \big]^T+q(1-q)[I-\tilde{P}F]\text{diag}(\Gamma_{k-1})[I-\tilde{P}F]^T, \label{second_a_k}  \\
 \Phi_{k+1}&=F\big[q^2 \Gamma_k +q(1-q) \text{diag}(\Gamma_k)\big]F^T,\label{second_y_k} \\
 \Psi_{k}&=\big[q\tilde{P}F+(1-q)I \big]\Psi_{k-1}\big[q\tilde{P}F+(1-q)I \big]^T+q(1-q)[I-\tilde{P}F]\text{diag}(\Psi_{k-1})[I-\tilde{P}F]^T,  \label{second_b_k}  \\
 \Lambda_{k+1}&=F\big[q^2 \Psi_k +q(1-q) \text{diag}(\Psi_k)\big]F^T, \label{second_z_k} \\
 \Xi_{k}&=\big[q\tilde{P}F+(1-q)I \big]\Xi_{k-1}\big[q\tilde{P}F+(1-q)I \big]^T+q(1-q)[I-\tilde{P}F]\text{diag}(\Xi_{k-1})[I-\tilde{P}F]^T,  \label{second_ab_k}  \\
 \Upsilon_{k+1}&=F\big[q^2 \Xi_k +q(1-q) \text{diag}(\Xi_k)\big]F^T  \label{second_yz_k},
\end{align}
with initial conditions
\begin{align}
& \Gamma_0=\tilde{P}y_0y_0^T\tilde{P}^T, \label{second_a_1} \\
& \Phi_1=\overline{y}_1\overline{y}_1^T+q(1-q) F\text{diag}(\tilde{P}y_0y_0^T\tilde{P}^T)F^T, \label{second_y_1} \\
& \Psi_0=\tilde{P}z_0z_0^T\tilde{P}^T, \label{second_b_1}  \\
& \Lambda_1=\overline{z}_1\overline{z}_1^T+q(1-q) F\text{diag}(\tilde{P}z_0z_0^T\tilde{P}^T)F^T, \label{second_z_1} \\
& \Xi_0=\tilde{P}y_0z_0^T\tilde{P}^T, \label{second_ab_1}   \\
& \Upsilon_1=\overline{y}_1\overline{z}_1^T+q(1-q) F\text{diag}(\tilde{P}y_0z_0^T\tilde{P}^T)F^T. \label{second_yz_1} 
%& \Phi_1=\overline{y}_1\overline{y}_1^T+q(1-q)\text{diag}(P\text{diag}(y_0 y_0^T)P^T),
\end{align}
%& \Phi_{k+1}=q^2\Gamma_{k+1}+q(1-q)\text{diag}(\Gamma_{k+1}) \label{covariance}.

\end{lemma}
\begin{proof}
The derivation of  \eqref{second_a_k},   \eqref{second_y_k}, \eqref{second_a_1}, and \eqref{second_y_1}, is the same as the derivation of   \eqref{second_b_k},   \eqref{second_z_k}, \eqref{second_b_1}, and \eqref{second_z_1}, thus the developments in the proof will only address the former. For $k=0$, it follows from Lemma~\ref{first_moment} and \eqref{ay_dynamics_1} that $a_0 =\tilde{P}y_0$. Then, 
\begin{align}
& \Gamma_0=\Ex[a_0a_0^T]=\tilde{P}\Ex[y_0y_0^T]\tilde{P}^T=\tilde{P} y_0y_0^T \tilde{P}^T, 
%& \Phi_1=\Ex[y_1y_1^T]=\Ex \big[[A_1 \circ X_1]uu^T[X_1 \circ A_1]^T\big] \label{Phi_1}.
\end{align}
and 
\begin{align}
& \Phi_1=\Ex[y_1y_1^T]=\Ex \big[F(a_0 \circ x_0)(x_0 \circ a_0)^TF^T\big]=F\Ex \big[(a_0 \circ x_0)(x_0 \circ a_0)^T\big]F^T \label{Phi_1}.
\end{align}
Applying the results in Lemmas~\ref{first_moment} and \ref{exp_lemma} to \eqref{Phi_1}, it follows that
\begin{align}
\Phi_1&=q^2F\Ex \big[ a_0  a_0^T \big]F^T+q(1-q)F\Ex \big[ \text{diag}(a_0   a_0^T) \big]F^T  \nonumber \\
&=(qF\tilde{P} y_0)(qF\tilde{P} y_0)^T+q(1-q)F\Ex[\text{diag}(\tilde{P}y_0y_0^T\tilde{P}^T) ]F^T \nonumber \\
&=(qP y_0)(qP y_0)^T+q(1-q) F\text{diag}(\tilde{P}y_0y_0^T\tilde{P}^T)F^T  \nonumber\\
& =\overline{y}_1\overline{y}_1^T+q(1-q) F\text{diag}(\tilde{P}y_0y_0^T\tilde{P}^T)F^T,
\end{align}
where we used the fact that $F\tilde{P}=P$ \big(refer to Eq.  \eqref{FtildeP}\big).

For $k \geq 1$, and taking into account that $y_k=F(a_{k-1} \circ x_{k-1})$, it follows that
\begin{align}
 \Gamma_{k} =& \Ex \big[\big(a_{k-1} \circ (u-x_{k-1}) +\tilde{P}y_k\big)\big(a_{k-1} \circ (u-x_{k-1})+\tilde{P}y_k\big)^T \big] \nonumber \\
=& \Ex \big[\big(a_{k-1} \circ (u-x_{k-1})\big)\big(a_{k-1} \circ (u-x_{k-1})\big)^T \big]+
\Ex \big[ \big(a_{k-1} \circ (u-x_{k-1})\big)\big(\tilde{P}y_k\big)^T  \big] \nonumber\\
&+\Ex \big[\big(\tilde{P}y_k\big)\big(a_{k-1} \circ (u-x_{k-1})\big)^T  \big]+\Ex \big[  \big(\tilde{P}y_k\big)\big(\tilde{P}y_k\big)^T  \big] \nonumber \\
=&  \Ex \big[\big(a_{k-1} \circ (u-x_{k-1})\big)\big(a_{k-1} \circ (u-x_{k-1})\big)^T \big]+\Ex \big[ \big(a_{k-1} \circ (u-x_{k-1})\big)\big(a_{k-1} \circ x_{k-1}\big)^T  \big]F^T\tilde{P}^T  \nonumber \\
&+\tilde{P}F\Ex \big[\big(a_{k-1} \circ x_{k-1}\big)\big(a_{k-1} \circ (u-x_{k-1})\big)^T  \big]+\tilde{P}F\Ex \big[ \big(a_{k-1} \circ x_{k-1}\big)\big(a_{k-1} \circ x_{k-1}\big)^T  \big]F^T\tilde{P}^T.  \label{Gamma_k+1a}
\end{align}
Then, from Lemma~\ref{exp_lemma},  \eqref{Gamma_k+1a} can be rewritten as
\begin{align}
 \Gamma_{k} = & (1-q)^2 \Ex \big[ a_{k-1}a_{k-1}^T \big] +q(1-q) \Ex \big[\text{diag}(a_{k-1}a_{k-1}^T)]  \nonumber \\
   &+q(1-q)\Ex \big[ a_{k-1} a_{k-1}^T \big]F^T\tilde{P}^T-q(1-q) \Ex \big[\text{diag}(a_{k-1}a_{k-1}^T)]F^T\tilde{P}^T  \nonumber\\
 &+q(1-q)\tilde{P}F\Ex \big[ a_{k-1}   a_{k-1}^T \big]-q(1-q)\tilde{P}F\Ex \big[\text{diag}(a_{k-1}a_{k-1}^T)] \nonumber\\
  &+ q^2\tilde{P}F\Ex \big[ a_{k-1} a_{k-1}^T \big]F^T\tilde{P}^T+q(1-q)\tilde{P}F\Ex \big[\text{diag}(a_{k-1}a_{k-1}^T)] F^T\tilde{P}^T.  \label{Gamma_k+1b}
 \end{align}
By re-arranging terms in \eqref{Gamma_k+1b} and observing that   $\Gamma_{k-1}=\Ex \big[ a_{k-1} a_{k-1}^T \big]$ and $\text{diag}(\Gamma_{k-1})=\Ex \big[ \text{diag}(a_{k-1}  a_{k-1}^T) \big]$, the result in \eqref{second_a_k} follows. 
 
Additionally, from Lemma~\ref{exp_lemma}, it follows that
 \begin{align}
 \Phi_{k+1}= \Ex \big[y_{k+1}y_{k+1}^T\big ] &= F\Ex \big[(a_k \circ x_k)(x_k \circ a_k )^T \big ]F^T  \nonumber \\
 &=F\Big[ q^2\Ex \big[ a_k a_k ^T \big ]+q(1-q)\Ex \big[\text{diag}(a_{k}a_{k}^T) \big] \Big]F^T \nonumber \\
 &=F\big[q^2\Gamma_{k}+q(1-q)\text{diag}(\Gamma_k) \big]F^T.
 \end{align}

To obtain the iterations for   $\Xi_k$ and $\Upsilon_k$, the developments are very similar to the ones above. For $k=0$, 
 \begin{align}
 \Xi_0=\Ex \big[ a_0   b_0^T \big]=\tilde{P}\Ex[y_0z_0^T]\tilde{P}^T=\tilde{P} y_0z_0^T \tilde{P},
 \end{align}
 and
 \begin{align}
 \Upsilon_1 =\Ex[y_1z_1^T]&=\Ex \big[F(a_0 \circ x_0)(x_0 \circ b_0)^TF^T\big]=F\Ex \big[(a_0 \circ x_0)(x_0 \circ b_0)^T\big]F^T \nonumber \\
&=(qF\tilde{P} y_0)(qF\tilde{P} z_0)^T+q(1-q)F\Ex[\text{diag}(\tilde{P}y_0z_0^T\tilde{P}^T) ]F^T \nonumber \\
&=(qP y_0)(qP z_0)^T+q(1-q) F\text{diag}(\tilde{P}y_0z_0^T\tilde{P}^T)F^T  \nonumber\\
&=\overline{y}_1\overline{z}_1^T+q(1-q) F\text{diag}(\tilde{P}y_0z_0^T\tilde{P}^T)F^T ,
\end{align}
where again we used the fact that $F\tilde{P}=P$ \big(refer to Eq.  \eqref{FtildeP}\big).

For $k \geq 1$, from Lemma~\ref{exp_lemma} and \eqref{ay_dynamics_1}, and taking into account that $y_{k+1}=F(a_k \circ x_k)$ and $z_{k+1}=F(b_k \circ x_k)$, it follows that
\begin{align}
 \Xi_{k} =& \Ex \big[\big(a_{k-1} \circ (u-x_{k-1}) +\tilde{P}y_k\big)\big(b_{k-1} \circ (u-x_{k-1})+\tilde{P}z_k\big)^T \big] \nonumber \\
=& \Ex \big[\big(a_{k-1} \circ (u-x_{k-1})\big)\big(b_{k-1} \circ (u-x_{k-1})\big)^T \big]+
\Ex \big[ \big(a_{k-1} \circ (u-x_{k-1})\big)\big(\tilde{P}z_k\big)^T  \big] \nonumber\\
&+\Ex \big[\big(\tilde{P}y_k\big)\big(b_{k-1} \circ (u-x_{k-1})\big)^T  \big]+\Ex \big[  \big(\tilde{P}y_k\big)\big(\tilde{P}z_{k}\big)^T  \big] \nonumber \\
=&  \Ex \big[\big(a_{k-1} \circ (u-x_{k-1})\big)\big(b_{k-1} \circ (u-x_{k-1})\big)^T \big]+\Ex \big[ \big(a_{k-1} \circ (u-x_{k-1})\big)\big(b_{k-1} \circ x_{k-1}\big)^T  \big]F^T\tilde{P}^T  \nonumber \\
&+\tilde{P}F\Ex \big[\big(a_{k-1} \circ x_{k-1}\big)\big(b_{k-1} \circ (u-x_{k-1})\big)^T  \big]+\tilde{P}F\Ex \big[ \big(a_{k-1} \circ x_{k-1}\big)\big(b_{k-1} \circ x_{k-1}\big)^T  \big]F^T\tilde{P}^T  \nonumber \\
=&   (1-q)^2\Ex \big[ a_{k-1}b_{k-1}^T \big]+q(1-q)\Ex \big[\text{diag}(a_{k-1}b_{k-1}^T)] \nonumber \\
  &+q(1-q)\Ex \big[ a_{k-1} b_{k-1}^T \big] F^T\tilde{P}^T-q(1-q) \Ex \big[\text{diag}(a_{k-1}b_{k-1}^T)]F^T\tilde{P}^T  \nonumber\\
 &+q(1-q)\tilde{P}F\Ex \big[ a_{k-1}   b_{k-1}^T \big]-q(1-q)\tilde{P}F\Ex \big[\text{diag}(a_{k-1}b_{k-1}^T)] \nonumber\\
  &+ q^2\tilde{P}F\Ex \big[ a_{k-1} b_{k-1}^T \big]F^T\tilde{P}^T+q(1-q)\tilde{P}F\Ex \big[\text{diag}(a_{k-1}b_{k-1}^T)] F^T\tilde{P}^T.\label{Gamma_k+1c}
\end{align}
By re-arranging terms in \eqref{Gamma_k+1c} and observing that   $\Xi_{k-1}=\Ex \big[ a_{k-1} b_{k-1}^T \big]$ and $\text{diag}(\Xi_{k-1})=\Ex \big[ \text{diag}(a_{k-1}  b_{k-1}^T) \big]$, the result in \eqref{second_ab_k} follows. Finally,
 \begin{align}
 \Upsilon_{k+1}= \Ex \big[y_{k+1}z_{k+1}^T\big ] &= F\Ex \big[(a_k \circ x_k)(x_k \circ b_k )^T \big ]F^T  \nonumber \\
 &=F\Big[ q^2\Ex \big[ a_k b_k ^T \big ]+q(1-q)\Ex \big[\text{diag}(a_{k}b_{k}^T) \big] \Big]F^T \nonumber \\
 &=F\big[q^2\Xi_{k}+q(1-q)\text{diag}(\Xi_k) \big]F^T,
 \end{align}
which completes the proof.
 \end{proof}
 
Although omitted in the statement of Lemma~\ref{second_moment}, it is easy to see that the dynamics of $\Delta_k=\Ex[b_ka_k^T]$ and $\Theta_k=\Ex[z_ky_k^T]$ can also be obtained by noting that $\Delta_k=\Psi_k^T$ and $\Theta_k=\Upsilon_k^T$.

 \section{Convergence Analysis of Robustified Double-Iteration Algorithm} \label{convergence_analysis}
 The previous Section established that the iterations governing the evolution of $\Gamma_k$, $\Psi_k$ and $\Xi_k$ are identical except for the initial conditions. We will show next that the steady-state solutions of these iterations are  also identical up to a multiplicative constant. To see this, we will rewrite \eqref{second_a_k}, \eqref{second_b_k}, and \eqref{second_ab_k} in vector form using  Kronecker products. For given matrices $C$, $A$, and $B$ of appropriate dimensions, the matrix equation $C=AXB$ (where $X$ is an unknown matrix) can be rewritten as a set of linear equations of the form $(B^T\otimes A)x=c$, where $x$ and $c$ are the vectors that result from stacking the columns of matrices $X$ and $C$ respectively, and $\otimes$ denotes the Kronecker product\footnote{The Kronecker product of matrices $A=[a_{ij}] \in \mathbb{R}^{m \times n}$ and $B=[b_{ij}] \in \mathbb{R}^{p \times q}$  is defined  (see, e.g.,  \cite{HoJo:91}) as the block matrix 
\begin{align}
A \otimes B := \begin{bmatrix}a_{11}B & \dots & a_{1n}B \\
\vdots & \ddots & \vdots \\
a_{m1}B & \dots & a_{mn}B \end{bmatrix} \in \mathbb{R}^{mp \times nq}. \nonumber
\end{align}
} of matrices \cite{HoJo:91}. Let $\gamma_k$ be the vector that results from stacking the columns of  $\Gamma_k$ and $\tilde{\gamma}_k$ the vector that results from stacking the columns of $\text{diag}(\Gamma_k)$. Then, it can be easily seen that \eqref{second_a_k} can be rewritten as
 \begin{align}
 \gamma_{k}=&\big[[q\tilde{P}F+(1-q)I] \otimes [q\tilde{P}F+(1-q)I]  \big]\gamma_{k-1}+\big[q(1-q)[I-\tilde{P}F] \otimes [I-\tilde{P}F] \big]\tilde{\gamma}_{k-1},~k \geq 1. \label{full_dynamics_a}
 \end{align}
 Let $G$ be a diagonal matrix with entries  $G\big((l-1)n^2+l,(l-1)n^2+l\big)=1,~\forall l=1,2,\dots,n^2$, and zero otherwise. Then, the second term on the right hand side of \eqref{full_dynamics_a} can be written as 
 \begin{align}
& q(1-q) \big( [I-\tilde{P}F] \otimes [I-\tilde{P}F] \big) \tilde{\gamma}_{k-1}= q(1-q)\big([I-\tilde{P}F] \otimes [I-\tilde{P}F] \big) G \gamma_{k-1}, ~k \geq 1,
 \end{align}
which leads us to
\begin{align}
& \gamma_{k}= \big[[q\tilde{P}F+(1-q)I] \otimes [q\tilde{P}F+(1-q)I]  +q(1-q)\big([I-\tilde{P}F] \otimes [I-\tilde{P}F] \big)G \big ] \gamma_{k-1},~k \geq 1.  \label{gamma_dyn}
\end{align}
Let $\psi_k$ and $\xi_k$ and $\delta_k$ be the vectors that result from stacking the columns of $\Psi_k$, $\Xi_k$ and $\Delta_k$ respectively. Then, it is easy to see that the same recurrence relation as in \eqref{gamma_dyn} governs the evolution of $\psi_k$ and $\xi_k$. 

\begin{theorem} \label{the_matrix_thm}
Let $P \in \mathbb{R}^{n \times n}$ be a column stochastic and primitive weight matrix associated with a   directed graph $\mathcal{G}=\{\mathcal{V},\mathcal{E}\}$, with $\mathcal{V}=\{1,2,\dots,n\}$ and $\mathcal{E}\subseteq \mathcal{V} \times \mathcal{V}$. Let $F=[I_n~ I_n~ \dots~ I_n] \in \mathbb{R}^{n \times n^2}$, where $I_n$ is the $n \times n$ identity matrix, and $\tilde{P}=[E_1P^T~E_2P^T~\dots~E_nP^T ]^T \in \mathbb{R}^{n^2 \times n}$,  where each $E_i \in \mathbb{R}^{n \times n},~i \in \{1,2,\dots,n \}$, satisfies $E_i(i,i)=1$  and has all other entries equal to zero. Then, for any $q$, $0 < q \leq 1$, the matrix $\Pi$ defined as
\begin{align}
\Pi \equiv [q\tilde{P}F+(1-q)I] \otimes [q\tilde{P}F+(1-q)I]  +q(1-q)\big([I-\tilde{P}F] \otimes [I-\tilde{P}F] \big)G   \label{THE_MATRIX}
 \end{align}
 \end{theorem}
is column stochastic, and it has a single  eigenvalue of maximum magnitude at value one.

\begin{proof}
We  show first column stochasticity of matrix $\Pi$. Let $C= q\tilde{P}F+(1-q)I$ and $D=  I-\tilde{P}F$, so that $\Pi=C \otimes C+q(1-q)(D \otimes D)G$. We will establish that $C \otimes C$
is column stochastic and also show that the column sums of $D \otimes D$ are all zero. By construction, the entries of the $i^{th}$ column of  $\tilde{P}\in \mathbb{R}^{n^2 \times n }$ are all zero, with the possible exception of the ones indexed by $\big((i-1)n+j,i\big),~i,j=1,2,\dots,n$, each of which corresponds to the $(j,i)$ entry of matrix $P$. Then,  it follows that $\sum_{l=1}^{n^2}\tilde{P}_{li}=\sum_{j=1}^{n}P_{ji}=1,\forall i=1,2,\dots,n^2$. The matrix $\tilde{P}F \in \mathbb{R}^{n^2 \times n^2}$  is also column stochastic by construction, as it results from horizontally concatenating $n$ times the matrix $\tilde{P}$, i.e., $\tilde{P}F=[\tilde{P}~ \tilde{P}~\dots~\tilde{P}]$; therefore, the matrix $C$ is also column stochastic. The kronecker product of $C$ with itself, results in an $n^4 \times n^4$ block matrix of the form  $C  \otimes C=[C_1~C_2 \dots~C_{n^2}]$, where $C_j=[c_{1j}C^T~c_{2j}C^T~\dots~c_{n^2j}C^T ]^T$. Then, it follows that  the sum of the entries of the $l^{th}$ column of $C_j$ is $\sum_{m=1}^{n^4} C_j(m,l)=(\sum_{i=1}^{n^2} c_{ij})(\sum_{r=1}^{n^2 }c_{rl})$. Since $ \sum_{i=1}^{n^2} c_{ij} $ and $\sum_{r=1}^{n^2 }c_{rl}$ are the sum of the entries of the $j^{th}$ and $l^{th}$ columns of $C=q\tilde{P}F+(1-q)I_{n^4}$ (which is column stochastic), it follows that  $\sum_{m=1}^{n^4} C_j(m,l)=1$; therefore,   $C \otimes C$ is also   column stochastic. 

Since $\tilde{P}F$ is column stochastic, the column-sums of $D=I- \tilde{P}F$ are zero. The kronecker product of D with itself is of the form $D  \otimes D=[D_1~D_2 \dots~D_{n^2}]$, where $D_j=[d_{1j}D^T~d_{2j}D^T~\dots~d_{n^2j}D^T ]^T$. Using  similar arguments as above, it follows that $\sum_{m=1}^{n^4} D_j(m,l)=(\sum_{i=1}^{n^2} d_{ij})(\sum_{r=1}^{n^2 }d_{rl})=0$, which implies that   the column-sums of $D \otimes D$ are zero. The only thing left to establish that $\Pi$ is column stochastic is to show that all entries of $\Pi$ are nonnegative (from where it immediately follows that $\Pi=C \otimes C+q(1-q)(D \otimes D)G$ is column stochastic). We argue nonnegativity of $\Pi$ as follows: due to the sparsity structure of $G$ in \eqref{THE_MATRIX}, the only nonzero entries of $(D \otimes D)G$ will be in columns $(k-1)n^2+k,~k=1,2,\dots,n^2$; thus except for  entries in these columns, the   entries  of $\Pi$ will be identical to the corresponding entries   in $C \otimes C$.  From the structure of $\tilde{P}F$,  entries of $C \otimes C$ and $q(1-q)(D \otimes D)$ can, respectively, take one of the following three forms:
\begin{align}
&  \big(qp_{ij} +(1-q)\big)\big(qp_{lm} +(1-q)\big), \label{c1} \\
&  q(1-q)(1-p_{ij})(1-p_{lm}),  \label{c2}
\end{align}
or 
\begin{align}
&   qp_{ij}  \big(qp_{lm} +(1-q)\big),  \label{c3}  \\
&  -q(1-q)p_{ij}(1-p_{lm}),   \label{c4}
\end{align}
or 
\begin{align}
&  q^2p_{ij}   p_{lm},   \label{c5}  \\
&  q(1-q) p_{ij} p_{lm},    \label{c6}
\end{align}
where $p_{ij}\geq0$ and $p_{lm}\geq0$ are the $(i,j)$ and $(l,m)$ entries of matrix $P$. For \eqref{c1} and \eqref{c2}, the corresponding entry of $\Pi$ is of the form 
\begin{align}
&  \big(qp_{ij} +(1-q)\big)\big(qp_{lm} +(1-q)\big)+q(1-q)(1-p_{ij})(1-p_{lm})=qp_{ij}p_{lm}+(1-q) , \label{pi1}
\end{align}
and satisfies $0 \leq qp_{ij}p_{lm}+(1-q) \leq 1$
For \eqref{c3} and \eqref{c4}, the corresponding entry of $\Pi$ is of the form 
\begin{align}
& qp_{ij}  \big(qp_{lm} +(1-q)\big)-q(1-q)p_{ij}(1-p_{lm})=qp_{ij}p_{lm} , \label{pi2}
\end{align}
and satisfies $0 \leq qp_{ij}p_{lm}\leq 1$. For \eqref{c5} and \eqref{c6}, the corresponding entry of $\Pi$ is of the form 
\begin{align}
q^2p_{ij}   p_{lm}+q(1-q) p_{ij} p_{lm}=qp_{ij}p_{lm}, \label{pi3}
\end{align}
and satisfies $0\leq qp_{ij}p_{lm} \leq 1$.

To prove the second assertion, we will show first that matrix $\tilde{P}F$ can be written via a permutation of its indices in the form  
\begin{align}
\begin{bmatrix} U && V \\ 0 && W \end{bmatrix}, \label{the_matrix_decomp}
\end{align}
where $U$ is an irreducible column stochastic matrix and $\lim_{k \rightarrow \infty} W^k =0$. Since $\tilde{P}F$ is column stochastic, we can assume that it corresponds to the weight matrix of   some   graph $\tilde{\mathcal{G}}=\{\tilde{\mathcal{V}}, \tilde{\mathcal{E}} \}$. We will show that this graph has a single recurrent class plus a few transient states, from which the decomposition of $\tilde{P}F$ in  \eqref{the_matrix_decomp}  follows. Let 
\begin{align}
& \tilde{\mathcal{V}}=\{(1,1), (2,1), \dots, (n,1),(1,2), (2,2), \dots,(n,2), \dots ,(n,n-1), (1,n),(2,n), \dots, (n,n) \}.
\end{align}
From the structure of $\tilde{P}F$, it follows that for any node $(i,j) \in \tilde{\mathcal{V}}$, one-step transitions out of $(i,j)$ are to nodes of the form $(m,i)$, with $i \in \mathcal{N}_m^-$, where $ \mathcal{N}_m^-$ is the set in-neighbors of  node $m$ in the graph $\mathcal{G}$  (with weight matrix $P$).  From the structure of $\tilde{P}F$, it also follows that there are possibly several rows of $\tilde{P}F$ with all  entries equal to zero, which means that a node  $(i,j)$ that is associated with such row  cannot be reached from any other node; however, as already argued, from nodes of the form $(i,j)$, it is possible to reach nodes of the form $(m,i)$, where $i \in \mathcal{N}_m^-$. Clearly, the nodes corresponding to rows with all entries being zero are transient.  Note that  the possibility of individual nodes that cannot be reached from any other node being disconnected is ruled out as it is easy to see the only nonzero diagonal entries of  $\tilde{P}F$ correspond to  diagonal entries of $P$, which are strictly smaller than one.

Next we will show that from a node $(i,j)$ whose corresponding row in $\tilde{P}F$  has some nonzero entries one can reach any other node $(m,l)$ whose corresponding  row  in $\tilde{P}F$ has some nonzero entries. This means that all non-transient nodes form a single recurrent class (as already argued all nonzero diagonal entries are strictly smaller than one which means there cannot be absorbing nodes). This  follows from the fact that the graph $\mathcal{G}$ is strongly connected, which means that for any $j,l \in \mathcal{V}$, there exists a path between $j$ and $l$. Let $i_1,i_2,\dots,i_t$ denote the nodes traversed along the path between $j$ and $l$. We will show next that for any two non-transient nodes $(i,j),  (r,l) \in \tilde{\mathcal{V}}$ there exists a path. As already argued, from $(i,j)$ one can reach in a single hop any node of the form $(m,i)$, where $m$ is a neighbor of node $i$ in the graph $\mathcal{G}$. Since $i_1$ is the first node traversed in the path between $j$ and $l$, it follows that $(i_1,i) \in \tilde{\mathcal{V}}$ can be reached in one step from $(i,j)$. By repeatedly using this  argument, it follows that the sequence of nodes $(i_1,i),(i_2,i_1),\dots,(i_{t},i_{t-1}), (r,i_t)$ forms a path between $ (i,j)$ and  $(r,l) $, which means that any non-transient node can be reached by any other non-transient node; thus, the set of non-transient nodes forms a single recurrent class. Clearly, the vertex set $\tilde{\mathcal{V}}$ can be decomposed into a single recurrent class and possibly several transient nodes. By re-ordering the nodes, it follows that $\tilde{P}F$ can be rewritten as in \eqref{the_matrix_decomp} (see, e.g., \cite[p. 126]{Se:06}). Furthermore, since $Q$ in \eqref{the_matrix_decomp} is irreducible, it follows that $ qQ+(1-q)I$ (where $I$ is the identity matrix) is primitive. It follows that $C=q\tilde{P}F+(1-q)I$ has a unique largest eigenvalue of value one,  i.e., $\lambda_1=1$, and $1>|\lambda_2|\geq \dots \geq|\lambda_{n^2}|$. Let $\sigma(C)=\{\lambda_1,\lambda_2,\dots,\lambda_{n^2}\}$. Then, $\sigma(C\otimes C )=\{\lambda_i \lambda_j,~i=1,\dots,n,~j=1,\dots,n \}$, including algebraic multiplicities in both cases  \cite[p. 245]{HoJo:91}. Since $\lambda_1=1$ is unique (multiplicity one) and $|\lambda_i| <1,~i=2,\dots,n^2$, it follows that the  eigenvalue of $C\otimes C =[q\tilde{P}F+(1-q)I] \otimes [q\tilde{P}F+(1-q)I]$ of largest magnitude also takes  value~1 and is unique. Since $C\otimes C$ is column stochastic and   $\lambda_1=1$ is unique, we know that  either $C\otimes C$ is also primitive  or it can be decomposed following a permutation of indices to the form \cite[p. 126]{Se:06}:
\begin{align}
\begin{bmatrix} L && M \\ 0 && N \end{bmatrix}, \label{the_matrix_decomp_2}
\end{align}
where $L$ is a primitive matrix and $\lim_{k \rightarrow \infty} N^k =0$. 

We will show next that  $\Pi=C \otimes C+q(1-q)(D \otimes D)G$ has exactly the same nonzero entries as $C \otimes C$ and therefore can be decomposed following the same permutation of indices to the form in \eqref{the_matrix_decomp_2}. As argued before, due to the sparsity structure of $G$ in \eqref{THE_MATRIX}, the only nonzero entries of $(D \otimes D)G$ will be in columns $(k-1)n^2+k,~k=1,2,\dots,n^2$, thus except for  entries in the aforementioned columns, the nonzero entries  of $\Pi$ will be the same as those in $C \otimes C$. For all other columns in $\Pi$ (that include nonzero entries in $(D \otimes D)G$), it was shown in 
\eqref{pi1}--\eqref{pi3} that the  nonzero entries of $\Pi$ are strictly positive, from where it follows  that $\Pi$ has the same sparsity structure as $C \otimes C$, which means that  $\Pi$ can also be decomposed in the form of  \eqref{the_matrix_decomp_2} (for some matrices $L'$, $M'$, $N'$), and the resulting upper-right block is also a primitive matrix. Therefore, $\Pi$ has a unique largest eigenvalue at one.
\end{proof}

The following two lemmas establish that the first and second moments of $a_k$ and $b_k$, and  $y_k$ and $z_k$ converge to the same solution up to a scalar multiplication. These two lemmas will be used to show that as $k \rightarrow \infty$, the random vector $v_k=z_k-\alpha y_k$, for $\displaystyle \alpha=\frac{\sum_{j=1}^n z_0(j)}{\sum_{j=1}^n y_0(j)}$, will converge almost surely to $v=0$. This suggests that, as $k \rightarrow \infty$, and whenever $z_k$ is nonzero, each node $i$ can obtain an estimate   of $\alpha=\frac{\sum_{j=1}^n z_0(j)}{\sum_{j=1}^n y_0(j)}$ by calculating the ratio $y_k(i)/z_k(i)$. We will also show that, in fact, $z_k$ will be larger than some threshold infinitely often.

%These two lemmas will be used to show that the as $k \rightarrow \infty$, the vectors $y_k$  and $z_k$ are linearly related, which explains why the ratio of $z_k$ and $y_k$ converges as $k \rightarrow \infty$.

~

~

~

~

~

\begin{lemma} \label{lemma_4}
The first moments of $a_k$  and $b_k$ (also $y_k$  and $z_k$ asymptotically converge to the same solution up to scalar multiplication:
\begin{align}
&  \lim_{k \rightarrow \infty}  \overline{z}_k=\alpha \lim_{k \rightarrow \infty}\overline{y}_k, \label{eq2} \\
&  \lim_{k \rightarrow \infty} \overline{b}_k=\alpha \lim_{k \rightarrow \infty}\overline{a}_k, \label{eq1} 
\end{align}
where $\displaystyle \alpha=\frac{\sum_{j=1}^n z_0(j)}{\sum_{j=1}^n y_0(j)}$.
\end{lemma}

\begin{proof}
In Lemma~\ref{first_moment}, it was shown that $\overline{y}_{k+1}=\big[qP+(1-q)I \big]\overline{y}_k$ and $\overline{z}_{k+1}=\big[qP+(1-q)I \big]\overline{z}_k$ with $\overline{y}_1=qy_0$, and  $\overline{z}_1=qz_0$. Since $P$ is column stochastic and primitive, it follows that $[qP+(1-q)I \big]$ is also column stochastic and primitive. Thus, $
\lim_{k \rightarrow \infty}\overline{z}_k=\alpha\lim_{k \rightarrow \infty}\overline{y}_k,$
where from the column stochasticity property it follows that $ \sum_{j=1}^n \overline{z}_k(j)= q(\sum_{j=1}^n z_0(j))$ and $ \sum_{j=1}^n \overline{y}_k(j)= q\big(\sum_{j=1}^n y_0(j)\big),~\forall k\geq 1$; this implies that $
\alpha=\frac{\sum_{j=1}^n z_0(j)}{\sum_{j=1}^n y_0(j)}$, 
which establishes \eqref{eq2}. 

By noting that $\overline{a}_{0}=\tilde{P}y_0$, and $\overline{b}_{0}=\tilde{P}y_0$, and using the fact that $\tilde{P}$ is column stochastic,  it follows that $\sum_{j=1}^n \overline{a}_k(j)= \sum_{j=1}^n \overline{a}_0(j) =\sum_{j=1}^n  y_0(j)$  and $\sum_{j=1}^n \overline{b}_k(j)= \sum_{j=1}^n \overline{b}_0(j)= \sum_{j=1}^n  z_0(j)$.  Since $q\tilde{P}F+(1-q)I $ (i.e.,  the matrix that governs the dynamics of $\overline{a}_k$ and $\overline{b}_k$)  is column stochastic and, as shown in the proof of Theorem~\ref{the_matrix_thm}, has a single largest eigenvalue at value 1, a similar development to the one above can be used to show  \eqref{eq1}.
\end{proof}

\begin{lemma} \label{w_dynamics}
Define $w_k=b_k-\alpha a_k$  and denote by $\chi_k$  the vector that results from stacking the columns of  $X_k:=\Ex[w_kw_k^T]$.  Then, it follows that
\begin{align}
\chi_{k}=\Pi \chi_{k-1}, \label{theta_dynamics}
\end{align}
\end{lemma}
with $\chi_0=\psi_0+\alpha^2\gamma_0-\alpha(\xi_0+\delta_0)$ and $\sum_{l=1}^{n^4}\chi_0(l)=0$.

\begin{proof}
Since $X_{k}:=\Ex[w_{k}w_{k}^T]=\Ex[b_{k}b_{k}^T]+\alpha^2\Ex[a_{k}a_{k}^T]-\alpha 
(\Ex[a_{k}b_{k}^T]+\Ex[b_{k}a_{k}^T])=\Psi_{k}+\alpha^2\Gamma_{k}-\alpha(\Xi_{k}+\Delta_{k})$, it follows that $\chi_{k}=\psi_{k}+\alpha^2\gamma_{k}-\alpha(\xi_{k}+\delta_{k})$. From \eqref{gamma_dyn} and subsequent discussion, it follows that $\gamma_k=\Pi \gamma_{k-1}$, $\psi_k=\Pi \psi_{k-1}$, $\xi_k=\Pi \xi_{k-1}$, and $\delta_k=\Pi \delta_{k-1}$, thus $\chi_{k}=\Pi \psi_{k-1}+\alpha^2\Pi \gamma_{k-1}-\alpha(\Pi \xi_{k-1}+\Pi \delta_{k-1})=\Pi(\psi_{k-1}+\alpha^2\gamma_{k-1}-\alpha(\xi_{k-1}+\delta_{k-1}))=\Pi\chi_{k-1}$.

In Lemma \ref{second_moment}, it was shown that $\Gamma_0=\tilde{P}y_0y_0^T\tilde{P}^T$, $\Psi_0=\tilde{P}z_0z_0^T\tilde{P}^T$, and $\Xi_0=\tilde{P}y_0z_0^T\tilde{P}^T=\Delta_0^T$. Since $\gamma_0$, $\psi_0$,  $\xi_0$, and $\delta_0$ result from stacking the columns of $\Gamma_0$, $\Psi_0$, $\Xi_0$, and $\Delta_0$, it follows that
\begin{align}
&\sum_{l=1}^{n^4}\gamma_0(l)=\sum_{i=1}^{n^2}\sum_{j=1}^{n^2}\Gamma_0(i,j)=\left(\sum_{i=1}^{n}y_0(i)\right)^2,  \\
&\sum_{l=1}^{n^4}\psi_0(l)=\sum_{i=1}^{n^2}\sum_{j=1}^{n^2}\Psi_0(i,j)=\left(\sum_{i=1}^{n}z_0(i)\right)^2,  \\
&\sum_{l=1}^{n^4}\xi_0(l)=\sum_{i=1}^{n^2}\sum_{j=1}^{n^2}\Xi_0(i,j)=\left(\sum_{i=1}^{n}y_0(i)\right)\left(\sum_{i=1}^{n}z_0(i)\right),  \nonumber \\
& \sum_{l=1}^{n^4}\delta_0(l)=\sum_{i=1}^{n^2}\sum_{j=1}^{n^2}\Delta_0(i,j)=\left(\sum_{i=1}^{n}z_0(i)\right)\left(\sum_{i=1}^{n}y_0(i)\right), \label{second_moment_init}
\end{align}
where the last equality   is obtained by taking into account that i)  matrix $\tilde{P}$ is column stochastic by construction, and ii) for any $a,b \in \mathbb{R}^n$, we have that $\sum_{i=1}^{n}\sum_{j=1}^{n}ab^T(i,j)=(\sum_{l=1}^{n}a_l)(\sum_{l=1}^{n}b_l)$.  Since $\alpha=\frac{\sum_{j=1}^n z_0(j)}{\sum_{j=1}^n y_0(j)}$, it follows that $\sum_{l=1}^{n^4}\chi_0(l)=\sum_{l=1}^{n^4}(\psi_0(l)+\alpha^2\gamma_0(l)-\alpha(\xi_0(l)+\delta_0(l)))=0$. \end{proof}

\begin{theorem}
\label{THEperfectlyaligned}
Let $y_k$ and $z_k$ be the random vectors  that result from iterations \eqref{ay_dynamics_1}--\eqref{ay_dynamics_2} and \eqref{by_dynamics_1}--\eqref{by_dynamics_2}. Define $v_k=z_k-\alpha y_k$, where $\alpha=\frac{\sum_{j=1}^n z_0(j)}{\sum_{j=1}^n y_0(j)}$. Then, $\|v_k\|_{\infty} \rightarrow 0$ almost surely. Furthermore, for every $j$, $v_k(j) \rightarrow 0$ as $k \rightarrow \infty$  almost surely (i.e., for every $j$, $\lim_{k \rightarrow \infty} v_k(j)=0$ with probability one).
%\begin{align}
%\lim_{k \rightarrow \infty} (z_k - \beta y_k) = 0 \; ,
%\end{align}
%where $\displaystyle \beta=\alpha=\frac{\sum_{j=1}^n z_0(j)}{\sum_{j=1}^n y_0(j)}$.
\end{theorem}

\begin{proof}
%We need to establish that $v_k(i)\rightarrow 0,~\forall i$, as $k \rightarrow \infty$ almost surely. 
The result   follows from the first Borel-Cantelli lemma \cite[Theorem 7.3.10]{Grimmett:1992}. For all $k\geq 0$ and all $\epsilon >0$, define the event $E_k(\epsilon)=\{\|v_k\|_{\infty}>\epsilon\}$. We will first establish an upper bound on $\sum_{k=0}^{\infty}\Pr\{E_k(\epsilon)\} $ by noting that $\Pr\{E_k(\epsilon)\}=\Pr\{\|v_k\|_{\infty}>\epsilon\} \leq \frac{\Ex\left [\|v_k\|_{\infty}\right ]}{\epsilon} $, thus   $\sum_{k=0}^{\infty}\Pr\{E_k(\epsilon)\} \leq \frac{1}{\epsilon} \sum_{k=0}^{\infty} \Ex\left [\|v_k\|_{\infty}\right ] \leq \frac{1}{\epsilon} \sum_{k=0}^{\infty} \Ex\left [\|v_k\|_{2}\right ]$. Note that $\Ex\left [\|v_k\|_{2}\right ]=(\Ex[v_k^Tv_k]))^{1/2}= (\text{trace}(\Ex[v_kv_k^T]))^{1/2}=(\text{trace}(\Ex[z_kz_k^T])+\alpha^2\text{trace}(\Ex[y_ky_k^T])-2\alpha\text{trace}(\Ex[y_kz_k^T])^{1/2}$. We will next show that $\Ex\left [\|v_k\|_{2}\right ] \rightarrow 0$ as $k \rightarrow \infty$ geometrically fast. Using Lemma~\ref{second_moment}, it can be established that $\Ex[v_kv_k^T]=\Ex[z_kz_k^T]+\alpha^2 \Ex[y_ky_k^T]-\alpha(\Ex[y_kz_k^T]+\Ex[z_ky_k^T]))=F\big[q^2X_{k-1}+q(1-q)\text{diag}(X_{k-1}) \big]F^T$ where $X_{k-1}=\Ex[w_kw_k^T]$ as defined in Lemma~\ref{w_dynamics}, thus the evolution of $\Ex[v_kv_k^T]$ is governed by the evolution of $X_{k-1}$ or by $\chi_{k-1}$ (the vector that results from stacking the columns of $X_{k-1}$). In Theorem \ref{the_matrix_thm}, we showed that $\Pi$ has a unique eigenvector (with all entries strictly positive) associated to the largest eigenvalue $\lambda_1=1$. Then, the solution of \eqref{theta_dynamics} is unique and equal to this eigenvector (up to scalar multiplication). Since  $\Pi$ is a column stochastic matrix, and Lemma~\ref{w_dynamics} established that $\sum_{l=1}^{n^4}\chi_0(l)=0$, it follows that  $\sum_{l=1}^{n^4}\chi_k(l)=0, k\geq 0$, and therefore $\lim_{k \rightarrow \infty } \chi_k(l)=0,~\forall l$. Additionally, it is well-known that the convergence of \eqref{theta_dynamics} is geometric with a rate of convergence given by the eigenvalue $\lambda_2$ of $\Pi$ with the second largest modulus, which satisfies $|\lambda_2|<\lambda_1=1$ (see, e.g., \cite{Se:06}). Thus, we have established that $\chi_k(l)\rightarrow 0,~\forall l,$ geometrically fast, from where it follows that all the entries of $\Ex[v_kv_k^T]$ go to zero also geometrically fast. Therefore, the $\text{trace}(\Ex[v_kv_k^T])$ also goes to zero geometrically fast, so that  $ \Ex\left [\|v_k\|_{2}\right ]$ also goes to $0$ geometrically fast. It immediately follows that $\sum_{k=0}^{\infty} \Ex\left [\|v_k\|_{2}\right ] < \infty$ and therefore $\sum_{k=0}^{\infty} \Pr\{\|v_k\|_{\infty}\geq \epsilon\}< \infty$. Then, from the first Borel-Cantelli lemma $\Pr\{\|v_k\|_{\infty}\geq \epsilon~\text{infinitely often}\}=0$ (or $\Pr\{\|v_k\|_{\infty}\geq \epsilon~\text{i.o.}\}=0$). Finally, since, for every $j$, $\|v_k\|_{\infty}\geq |v_k(j)|$,  then, for every $j$, $\Pr\{\|v_k\|_{\infty}\geq \epsilon\}\geq \Pr\{|v_k(j)|\geq \epsilon\}$, and thus, for every $j$,  $\sum_{k=0}^{\infty} \Pr\{|v_k(j)|\geq \epsilon\}< \infty$. Then, by Theorem 7.2.4.c of \cite{Grimmett:1992}, for every  $j$, $v_k(j)\rightarrow 0$ almost surely. \end{proof}

Theorem~\ref{THEperfectlyaligned} has established that, in the limit as the number of iterations $k$ becomes large, the values of vectors $y_k$ and $z_k$ will be perfectly aligned so that $z_k - \alpha y_k = 0$ with probability one. Thus, in this limiting case, each node $j$ can calculate the value of $\frac{1}{\alpha}$ by taking the ratio $\frac{y_k(j)}{z_k(j)}$, {\em as long as} $z_k(j) \neq 0$. Note that, as also evidenced by the simulations provided for the small network of Fig.~\ref{FIGsmallgraph} (e.g., the plots on the left and in the middle for Figure~\ref{FIGsmallplotsq5}), the vectors $y_k$ and $z_k$ do not converge in any way;\footnote{Earlier, we established that, for large $k$, the quantities $E[y_k]$, $E[z_k]$, $E[y_k y^T_k]$, $E[z_k z^T_k]$ and $E[y_k z^T_k]$ converge, but this does not imply any convergence for the values of $y_k$ or $z_k$.} however, the values $y_k$ and $z_k$ become perfectly aligned (with probability one), allowing each node $j$ to calculate $\frac{1}{\alpha} = \frac{y_k(j)}{z_k(j)}$. The only problem here arises when $y_k(j)$ and $z_k(j)$ have both value zero, which does not constitute a violation of $z_k - \alpha y_k = 0$, but clearly does not allow node $j$ to calculate the desired value $\frac{1}{\alpha}$. This is evidenced also in the simulations provided for the small network of Fig.~\ref{FIGsmallgraph}: for example, in the plots in Fig.~\ref{FIGsmallplotsq9}, the values of $y_k(j)$ and $z_k(j)$ often go to zero (simultaneously) leaving their ratio undefined.\footnote{Since in the simulations for the plots in Fig.~\ref{FIGsmallplotsq9}, each packet (including self-packets) can be dropped with probability $1-q$ at iteration $k$, there is a nonzero probability that all packets destined for node $j$ will be dropped, causing both of its values at the next iteration ($y_{k+1}(j)$ and $z_{k+1}(j)$) to be zero. For instance, in the simulation of Fig.~\ref{FIGsmallplotsq9}, $z_k(1)$ will be zero with probability at least $(1-q)^2=0.81$ because node $1$ will have value zero if both packets destined for it (including the self-packet) are dropped.} The next two theorems essentially establish  that $z_k(j),~j=1,2,\dots,n$, will be greater than zero (in fact, greater than a constant $C$ that will be specified) infinitely often. Note that, in  subsequent developments, $z_k(j)$ is denoted with $z_j[k]$ in order to remain close to the notation in \eqref{iter_2_gen_1}--\eqref{iter_2_gen_3}.

%\textcolor{red}{this next theorem establishes the fact that there is a finite probability at each time instant $k$ of a particular $z_j[k]$ being greater than some threshold. I need a corollary or another theorem afterwards defining $\Pr\{E_k | \zeta_{k-1}, \dots,\zeta_2, \zeta_{1}\}$, where $\zeta_{i}$ is defined as the occurrence of $E_i$, where $E_i=\{z[ni]>C\},~i=1,2\dots$}

\begin{theorem}
\label{THEnonzero}
Consider a (possibly directed) strongly connected graph $\mathcal{G} = (\mathcal{V},\mathcal{E})$ and the iteration in \eqref{iter_2_gen_1}--\eqref{iter_2_gen_3}, where $x_{ji}[k]$, $(j,i) \in \mathcal{E}$, $k=0, 1, 2, ...,$ are independent identically distributed (i.i.d.) indicator R.V.'s as defined in \eqref{Bernouilli_model}, i.e., $x_{ji}[k] = 1$ with probability $q$ and $x_{ji}[k] = 0$ with probability $1-q$, independently between $(j,i) \in \mathcal{E}$ and independently for different $k$. For every $j=1,2,\dots,n$, define the event  $E_k^j=\{z_{j}[kn]\geq C\},~k \geq 1$, where $C = \frac{n}{(n+m)(\mathcal{D}^+_{\max})^{n-1}}$, $\mathcal{D}^+_{\max} = \max_{j \in \mathcal{V}} \{ \mathcal{D}^+_j \}$, $n = | \mathcal{V} |$, and $m = | \mathcal{E} |$. Let $\zeta_k^j$ denote  the indicator of the event $E_{k }^j,~k\geq 1$, i.e., $\zeta_k^j=1$ whenever $E_{k }^j,~k\geq 1$ occurs, and $\zeta_k^j=0$ otherwise. Then, whatever $\zeta_1, \zeta_2 , \dots, \zeta_{k-1}$,  we have that
\begin{align}
 \Pr\{z_{j}[(k+1)n]\geq C~|~\zeta_{k}^j,\zeta_{k-1}^j,\dots,\zeta_{1}^j\} \geq q^n, ~\forall j.
\end{align}
%
%
%For $k \geq n$, we have
%$$
%Pr \left [ z_j[k] \geq C \right ] > q^n \; ,
%$$
\end{theorem}

\begin{proof}
Note that the iteration in \eqref{iter_2_gen_1} to \eqref{iter_2_gen_3} involves nonnegative quantities: since for every $j$, $z_j[0]>0,~\forall j$, it follows from \eqref{by_dynamics_1}--\eqref{by_dynamics_2} that, for every $j$,  $z_j[k]\geq0,~k\geq 0$. Then, it is not hard to establish that the total mass\footnote{This notion is discussed in great detail in Part II of this paper.} $\mathcal{M}_{k+1}$ in the system, defined as 
\begin{equation}
\label{defMass}
\mathcal{M}_{k+1} := \sum_{j=1}^n z_j[k+1] + \sum_{(j,i) \in \mathcal{E}} (\sigma_{ji}[k] - \tau_{ji}[k-1])(1-x_{ji}[k]) \; ,
\end{equation}
satisfies
$$
\mathcal{M}_{k+1} = n \; , \mbox{ for all } k = 0, 1, 2, ...\; .
$$
[This follows from the fact that $M_0 = \sum_{j=1}^n z_j[0] = n$ and the observation that
\begin{eqnarray*}
\mathcal{M}_{k+1} & := & \sum_{j=1}^n z_j[k+1] + \sum_{(j,i) \in \mathcal{E}} (\sigma_{ji}[k] - \tau_{ji}[k-1]) (1-x_{ji}[k]) \\
          & = & \sum_{(j,i) \in \mathcal{E}} (\sigma_{ji}[k] - \tau_{ji}[k-1]) x_{ji}[k] + \sum_{(j,i) \in \mathcal{E}} (\sigma_{ji}[k] - \tau_{ji}[k-1]) (1-x_{ji}[k]) \\
          & = & \sum_{(j,i) \in \mathcal{E}} (\sigma_{ji}[k] - \tau_{ji}[k-1]) \\
          & = & \sum_{(j,i) \in \mathcal{E}} \left ( \sigma_{ji}[k-1] + \frac{1}{\mathcal{D}^+_j}z_j[k] - \sigma_{ji}[k-1]x_{ji}[k-1] - \tau_{ji}[k-2](1-x_{ji}[k-1]) \right ) \\
          & = & \sum_{j=1}^n z_j[k] + \sum_{(j,i) \in \mathcal{E}} (\sigma_{ji}[k-1] - \tau_{ji}[k-2] ) (1 - x_{ji}[k-1]) \; ,
\end{eqnarray*}
which is equal to $\mathcal{M}_k$.]

The definition of $\mathcal{M}_{k+1}$ in \eqref{defMass} involves the summation of $n+m$ {\em nonnegative} quantities, namely, $z_j[k+1]$ for $j=1, 2, ..., n$ and $m_{ji}[k+1] := (\sigma_{ji}[k] - \tau_{ji}[k-1]) (1-x_{ji}[k])$ for $(j,i) \in \mathcal{E}$. We can think of these quantities as follows: $z_j[k+1]$ is the mass at node $j$, whereas $m_{ji}[k+1]$ is the mass waiting to get transferred to node $j$ from node $i$. Since all of these quantities are nonnegative, at least one of them is larger or equal to $\frac{n}{n+m}$. Regardless of whether this quantity is associated with a node (say node $j^*$) or a link (say link $(j^*,i^*)$), this mass has at least one way of reaching any node $i$ of interest in graph $\mathcal{G}$ via a path of length at most $n-1$ (because the graph $\mathcal{G}$ is strongly connected): in particular, there is at least one path of length at most $n-1$ from node $j^*$ to node $i$ and all the links in this path have weight at least $\frac{1}{\mathcal{D}^+_{\max}}$. If all these links are activated, which occurs with probability $q^{n-1}$ ($q^n$ in the case of link $(j^*, i^*)$ because the mass needs to first transfer to $j^*$), then a fraction $(\frac{1}{\mathcal{D}^+_{\max}})^{n-1}$ of the mass will transfer to node $i$ in at most $n$ steps. Then, since for every $j$,  $z_j[k]\geq0,~k\geq 0$, independently of the values of $z_j[ln],~l=1,2,\dots,k$,   $\Pr\{z_j[(k+1)n]\geq C~|~\zeta_{k}^j,\zeta_{k-1}^j,\dots,\zeta_{1}^j\}\geq q^n$ obtains, whatever $\zeta_1, \zeta_2 , \dots, \zeta_{k-1}$. Finally,  for every $j$ , $\Pr\{z_{j}[(k+1)n]\geq C~|~\zeta_{k}^j,\zeta_{k-1}^j,\dots,\zeta_{1}^j\}=1-\Pr\{z_{j}[(k+1)n]< C~|~\zeta_{k}^j,\zeta_{k-1}^j,\dots,\zeta_{1}^j\}\leq 1-Pr\{z_{j}[(k+1)n]= 0~|~\zeta_{k}^j,\zeta_{k-1}^j,\dots,\zeta_{1}^j\} \leq 1-q^{\mathcal{D}_j^-}$, where $\mathcal{D}_j^-$ is the in-degree of node $j$.
%Thus, for any two instants $kn,$ $(k+1)n,~k= 1\dots l$, it follows that $\Pr\{z_j[k+1]-z_j[k]\geq C\}\geq q^n$. Then, since, for every $j$, $z_j[0]>0,~\forall j$, it follows from \eqref{by_dynamics_1}--\eqref{by_dynamics_2} that, for every $j$,  $z_j[k]\geq0,~k\geq 0$. Therefore $\Pr\{z_j[k+1]-z_j[k]\geq C\}=\Pr\{z_j[k+1]\geq C|\omega_k\}\geq q^n$
\end{proof}

%\begin{corollary}
%asdf
%\end{corollary}
%\begin{proof}
%For every $j$, following the same argument as in the proof of Theorem~\ref{THEnonzero}, independently of the values of $z_j[kn],~k= 1\dots l$, the probability that node $j$ will receive a mass greater or equal than $C$ at instant $(l+1)n$ is $q^n$. Thus, for any two instants $kn,$ $(k+1)n,~k= 1\dots l$, it follows that $\Pr\{z_j[k+1]-z_j[k]\geq C\}\geq q^n$. Then, since, for every $j$, $z_j[0]>0,~\forall j$, it follows from \eqref{by_dynamics_1}--\eqref{by_dynamics_2} that, for every $j$,  $z_j[k]\geq0,~k\geq 0$. Therefore $\Pr\{z_j[k+1]-z_j[k]\geq C\}=\Pr\{z_j[k+1]\geq C|\omega_k\}\geq q^n$
%
%
%at instant $n,2n,3n,\dots$, each node $j$ will have a positive value $z_j[kn],~k=1,2,\dots$ larger than $C$ with probability $q^n$$z_j[kn],~k=1,2,\dots$. 
%\end{proof}

Given a sequence of events $E_1,~E_2,\dots,E_n,\dots$ defined on some probability space, the next theorem (which we do not prove) states the  1912 Borel criterion for establishing whether the event that infinitely many of the $E_k$ occur, denoted by $\{E_k~\text{i.o}\}$, will occur  with probability one or zero  (see, e.g., \cite{Bo:12,Nash:1954}). This result, together with the result in Theorem~\ref{THEnonzero} will be used to establish that, for every $j$,  the event  $E_k^j=\{z_{j}[kn]\geq C\},~k \geq 1$ occurs  infinitely often.

\begin{theorem}\label{Borel_cantelli}
Let $\{E_k\},~k=1,2,\dots$, be a sequence of events defined on some probability space. Let $\zeta_k$ be the indicator function of the event $E_k$. Let $\Pr\{E_{k+1}~|~\zeta_{k},\zeta_{k-1},\dots,\zeta_{1}\}$ denote the conditional probability of the event $E_{k+1}$ given the outcome of previous trials. If $0<p_{k }' \leq  \Pr\{E_{k+1}~|~\zeta_{k },\zeta_{k-1},\dots,\zeta_{1}\} \leq p_{k}''$ for every $k$, whatever $\zeta_1, \zeta_2 , \dots, \zeta_{k}$, then i) $\Pr\{E_k~\text{i.o.}\}=0$ if $\sum_{k=1}^{\infty}p_k'' \leq \infty$, and ii) $\Pr\{E_k~\text{i.o.}\}=1$ if $\sum_{k=1}^{\infty}p_k'  = \infty$.
\end{theorem}

\begin{theorem}
Consider a (possibly directed) strongly connected graph $\mathcal{G} = (\mathcal{V},\mathcal{E})$ and the iteration in \eqref{iter_2_gen_1}--\eqref{iter_2_gen_3}. For every $j=1,2,\dots,n$, define the event  $E_k^j=\{z_{j}[kn]\geq C\},~k \geq 1$, where $C = \frac{n}{(n+m)(\mathcal{D}^+_{\max})^{n-1}}$, $\mathcal{D}^+_{\max} = \max_{j \in \mathcal{V}} \{ \mathcal{D}^+_j \}$, $n = | \mathcal{V} |$, and $m = | \mathcal{E} |$.  Then, $\Pr\{E_k~\text{i.o.}\}=1$.
\end{theorem}

\begin{proof}
Theorem~\ref{THEnonzero} established that, for  every $j$, $\Pr\{z_{j}[(k+1)n]\geq C~|~\zeta_{k}^j,\zeta_{k-1}^j,\dots,\zeta_{1}^j\}\geq q^n$. Define $p_{k}' =q^n$, then it follows that $\sum_{k=1}^{\infty}p_k'  = \infty$, and by the second assertion in Theorem ~\ref{Borel_cantelli}, we conclude that, for every $j$, $\Pr\{E_k^j~\text{i.o.}\}=1$.
\end{proof}

The final piece is to establish  that whenever $z_j[k]\geq C$, which occurs infinitely often, each node will be able to calculate an estimate of   $\overline{v}$ by calculating the ratio $y_j[k]/z_j[k]$ and this estimate will converge to $1/\alpha$ as $k$ goes to infinity.

\begin{theorem}
For each $j$, let $k=t_1,t_2,\dots  $ be an increase sequence of  time steps for which $z_j[k]>C$. Then, almost surely 
\begin{align}
\lim_{n \rightarrow \infty}\left|\frac{y_j[t_n]}{z_j[t_n]}-\frac{1}{\alpha}\right|=0.
\end{align}
\end{theorem}

\begin{proof}
Since $z_j[k]\geq C$ for $k=t_1,t_2,\dots $,  it follows that  $\frac{y_j[t_n]}{z_j[t_n]}-\frac{1}{\alpha}\leq \frac{\alpha y_j[t_n]-z_j[t_n]}{\alpha C} $. Also, in the proof of Theorem \ref{THEnonzero}, we established that $\mathcal{M}_k=n, k \geq 0$, from where it follows that   $z_j[t_n] \leq n$, therefore $\frac{y_j[t_n]}{z_j[t_n]}-\frac{1}{\alpha}\geq  \frac{\alpha y_j[t_n]-z_j[t_n]}{\alpha n}$.  In Theorem \ref{THEperfectlyaligned}, we established that $|\alpha y_j[k]-z_j[k]|\rightarrow 0$ almost surely, which implies that  the subsequence $|\alpha y_j[t_n]-z_j[t_n]|\rightarrow 0$ almost surely, then  since $C<n$, we have that 
\begin{align}
\lim_{n \rightarrow \infty}  \left|\frac{y_j[t_n]}{z_j[t_n]}-\frac{1}{\alpha}\right|\leq \lim_{n \rightarrow \infty} \left| \frac{\alpha y_j[t_n]-z_j[t_n]}{\alpha C} \right|=0
\end{align}
almost surely.
\end{proof}

\section{Concluding Remarks} \label{concluding_remarks}
In this paper, we proposed a method to ensure robustness of a class of linear-iterative distributed algorithms against unreliable communication links that may drop packets. We used statistical-moment analysis and the Borel-Cantelli lemmas to establish the correctness of the proposed robustified algorithm. In Part II of this paper, we establish similar convergence properties by recasting the problem as a finite  inhomogeneous Markov and using  coefficients of ergodicity commonly to used in analyzing this type of Markov chains.

%\begin{theorem}
%The second moments of $y_k$ and $z_k$ satisfy the following relation
%\begin{align}
%&\lim_{k \rightarrow \infty} \big(\Ex [y_kz_k^T]-\Ex[y_k]\Ex[z_k] \big) \big(\Ex [y_ky_k^T]-\Ex[y_k]\Ex[y_k] \big)^{-1} \big(\Ex [y_kz_k^T]-\Ex[y_k]\Ex[z_k] \big) \nonumber \\
%& =\lim_{k \rightarrow \infty} \big(\Ex [z_kz_k^T]-\Ex[z_k]\Ex[z_k] \big). \label{fundamental_result}
%\end{align}
%\end{theorem}
%\begin{proof}
%By  unwrapping both sides of \eqref{fundamental_result} and using the results in Lemmas \ref{lemma_4} and  \ref{lemma_5}, the result follows.
%\end{proof}

\bibliographystyle{IEEEtran}

\bibliography{distributed,distributed2}

\end{document}